\documentclass{article}
\usepackage{graphicx}
\usepackage[hidelinks]{hyperref}
\usepackage{doi}
\usepackage{caption}
\usepackage{subcaption}
\usepackage{float} 
\usepackage{xcolor}
\usepackage{amsmath}
\usepackage{amsthm}
\newtheorem{theorem}{Theorem}[section]
\newtheorem{corollary}[theorem]{Corollary}
\newtheorem{lemma}[theorem]{Lemma}
\newtheorem{definition}[theorem]{Definition}
\newtheorem{observation}[theorem]{Observation}
\usepackage{enumitem}
\usepackage{booktabs} 
\usepackage[normalem]{ulem}

\usepackage[dvipsnames]{xcolor}

\newcommand{\len}{\mathrm{len}}

\usepackage{algorithm}
\usepackage{algpseudocode}
\usepackage[numbers, sort]{natbib}

\newcommand{\OPT}{{\mathsf{OPT}}}

\newcommand{\cost}{{\mathsf{cost}}}

\newenvironment{obsproof}{\par\noindent\underline{\textit{Proof.}}}{\mbox{}\hfill$\bullet$\par}

\title{Minimizing Total Travel Time for Collaborative Package Delivery with Heterogeneous Drones}
\author{Thomas Erlebach\thanks{Department of Computer Science, Durham University, UK} \and Kelin Luo\thanks{Department of Computer Science and Engineering, University at Buffalo, USA} \and Wen Zhang\footnotemark[\value{footnote}]}
\date{}

\begin{document}

\maketitle

 \begin{abstract}
     Given a fleet of drones with different speeds and
     a set of package delivery requests, the collaborative
     delivery problem asks for a schedule for the
     drones to collaboratively carry out all package deliveries,
     with the objective of minimizing the total travel time
     of all drones. We show that the best non-preemptive
     schedule (where a package that is picked up at its source
     is immediately delivered to its destination by one drone)
     is within a factor of three of the best preemptive schedule
     (where several drones can participate in the delivery of
     a single package). Then, we present a constant-factor
     approximation algorithm for the problem of computing the
     best non-preemptive schedule. The algorithm reduces
     the problem to a tree combination problem and uses
     a primal-dual approach to solve the latter. We have
     implemented a version of the algorithm optimized for
     practical efficiency and report the results of experiments
     on large-scale instances with synthetic and real-world data,
     demonstrating that our algorithm is scalable and delivers schedules
     of excellent quality.   
 \end{abstract}
 
\section{Introduction}
\label{sec:intro}

The rise of autonomous vehicle technology has profoundly transformed the logistics and delivery industry, enabling faster and more efficient package delivery. Given the varying serving speeds, consumption costs, and service areas of different vehicles or drones, one of the key challenges in this domain is optimizing the delivery process, particularly when multiple vehicles are available to complete pickup and delivery tasks~\cite{savelsbergh201650th}.
This class of problems, known as Collaborative Delivery Problems (CD Problems), focuses on finding an optimal schedule to transport packages from one location to another through the collective effort of multiple vehicles.
There are two modes in which vehicles can collaborate on completing the pickup
and delivery tasks: In \emph{coarse-grained} collaboration, the vehicles collaborate
by partitioning the package delivery requests among themselves: Each package delivery request
is assigned to one vehicle, and each of the vehicles independently follows a route based on its assigned requests. In \emph{fine-grained} collaboration,
multiple vehicles can participate in the fulfillment of a single request:
The vehicles can collaborate by handing off packages between agents to enhance efficiency and reduce delivery time.

In this paper, we study the multi-package delivery problem in a collaborative setting.
We extend the classical problem of multiple vehicle pickup to a setting where the available vehicles have different efficiencies, particularly in terms of speed. We formulate this problem as a Heterogeneous Collaborative Delivery Problem and investigate how delivery planning can be optimized to minimize the objective of total travel time.
We use the terms `drone' and `vehicle' interchangeably in the remainder of the
paper.

Our contributions can be summarized as follows: (1)
We analyze the effectiveness of fine-grained collaboration among vehicles. Our results demonstrate that the benefit of fine-grained collaboration is limited. Specifically, we show that the gap between serving packages with fine-grained collaboration (possibly involving multiple vehicles in the delivery of a single package) and serving them with coarse-grained collaboration (serving each package with a single vehicle) is bounded by a constant factor.
(2) The main contribution of this paper is a constant-factor approximation algorithm
for the Heterogeneous CD problem with coarse-grained collaboration. We reduce the problem to a Tree Combination Problem and adapt a primal-dual algorithm presented for the multiple-vehicle TSP to address the latter problem. Our model accounts for the heterogeneous capabilities of vehicles and achieves a constant approximation ratio for minimizing total travel time (or total cost). Furthermore, our reduction shifts the problem complexity to depend primarily on the number of vehicles rather than the number of requests, which is typically much larger than the number of vehicles. This advancement offers a practical framework for optimizing real-world routing problems, where variations in vehicle capabilities play a critical role.
(3) We describe an optimized implementation of our algorithm and report
experimental results regarding running time and solution quality for a range
of synthetic and real-world instances. Compared to a baseline heuristic without
approximation guarantee, our algorithm produces solutions of comparable and
sometimes substantially better quality while displaying a significantly reduced running time.

The remainder of the paper is structured as follows. Section~\ref{sec:related} discusses related work. Section~\ref{sec:Preliminary} introduces relevant notation and definitions.
Section~\ref{sec:FromNonpreToPre} proves that the gap between fine-grained and
coarse-grained schedules is bounded. In Section~\ref{sec:Non-preemptiveMin-sumCD},
we present our constant-factor approximation algorithm for the heterogeneous CD
problem with coarse-grained collaboration. In Section~\ref{sec:adjustments}, we
describe the optimizations we have applied for the implementation of the algorithm.
Then, Section~\ref{sec:exp} presents the experimental results. Conclusions are
given in Section~\ref{sec:conclusion}.

\subsection{Related Work}\label{sec:related}

\paragraph{Collaboration on a Single Request.} 
The problem of delivering a single package has been well studied~\cite{bartschi_et_al:LIPIcs.STACS.2017.10, bartschi_et_al:LIPIcs.MFCS.2018.56, carvalho2021fast}.  
These studies consider a problem where \( k \) agents, initially located at distinct nodes of an undirected graph with \( n \) nodes and weighted edges, must collaboratively deliver a single item from a source node \( s \) to a target node \( t \). The agents can traverse the graph edges starting at time \( 0 \), with each agent \( i \) moving at a specific velocity \( \nu_i \) and energy
consumption rate \( w_i \).
The objective of minimizing total consumption cost was first studied by Bärtschi et al.~\cite{bartschi_et_al:LIPIcs.STACS.2017.10}, who provided an \( O(n^3) \) algorithm to find the optimal solution. The objective of minimizing delivery time was first addressed by Bärtschi et al.~\cite{bartschi_et_al:LIPIcs.MFCS.2018.56}, who proposed an algorithm with time complexity \( O(k^2 m + k n^2 + \text{APSP}) \), where APSP refers to the time required to compute all-pairs shortest paths in a graph with \( n \) nodes and \( m \) edges. Later, Carvalho et al.~\cite{carvalho2021fast} improved the time complexity for this problem to \( O(k n \log n + k m) \). All these results are based on the common assumption that every agent can travel freely throughout the entire graph without any restrictions. 

Studies on collaboration for a single request also explore scenarios where each agent's travel distance is limited~\cite{chalopin2014data, chalopin2014data2}, where an agent's overall movement is restricted~\cite{erlebach_et_al:LIPIcs.ISAAC.2022.49}, or even where the package's path is restricted~\cite{CHALOPIN202187}. The problem in which each agent has an energy budget that constrains its total distance traveled has been shown to be strongly NP-hard in general graphs~\cite{chalopin2014data}. Scenarios where each agent is constrained to a specified movement subgraph have been studied for both minimizing total consumption and total delivery time, and it has been shown to be NP-hard~\cite{erlebach_et_al:LIPIcs.ISAAC.2022.49}. The variant in which the package must travel via a fixed path in a general graph has been studied by Chalopin et al.~\cite{CHALOPIN202187}. 

\paragraph{Collaboration on Two or More Requests.}  
Although significant progress has been made, particularly with the proposal of fast polynomial-time algorithms, in understanding collaborative delivery in single-package scenarios, the problem becomes computationally challenging once multiple packages are introduced~\cite{bartschi_et_al:LIPIcs.STACS.2017.10}. In fact, even with just two packages, the minimum delivery time collaborative delivery problem, without any additional restrictions, is already NP-hard~\cite{carvalho2021fast}.

\paragraph{Dial a Ride.} 
The Dial-a-Ride Problem (DARP)—in both its preemptive and non-preemptive forms—involves a set of objects (each specified by a source-destination node pair) and a fleet of vehicles. The objective is to compute the shortest possible tour such that all objects are picked up from their sources and delivered to their destinations, without exceeding the vehicle’s capacity at any point in time. This problem is typically studied under capacity constraints, with the common assumption that all vehicles are homogeneous.  
Our problem can be viewed as a heterogeneous, unit-capacity, non-preemptive DARP.
While Frederickson et al.~\cite{frederickson1976approximation} gave a 1.8-approximation for the unit-capacity case on general graphs, the heterogeneous setting, to the best of our knowledge, remains unexplored.

\paragraph{Heterogeneous TSP.}   
A special case of our problem, where the pickup location equals the drop-off location, can be viewed as a Multiple Depot Heterogeneous Traveling Salesman Problem (MDHTSP), since each package delivery is considered completed as soon as its corresponding node is visited, with
only the minor difference that vehicles need to return to their depots in the MDHTSP. Rathinam et al.~\cite{rathinam2020primal} present a 2-approximation algorithm for MDHTSP based on the primal-dual method. Our algorithm also leverages the primal-dual method, with an extension of their method.

\section{Preliminaries}
\label{sec:Preliminary} 
 We are given a metric space \((X, d)\), where \(d: X \times X \to \mathcal{R}_{\geq 0}\) is a distance function satisfying the triangle inequality. A set of vehicles (or drones) \(K\) with \(k = |K|\) is specified, where each drone \(j \in K\) has a depot (or initial location) \(r_j \in X\) and a speed \(p_j\). To traverse from node \(u \in X\) to node \(w \in X\), drone \(j\) takes time \(\frac{d(u, w)}{p_j}\). Without loss of generality, we assume that for any two drones \(i\) and \(j\) with \(i < j\), the speed satisfies \(p_i \geq p_j\).  
We are also given a set of \(m\) packages to be delivered, each specified by a pickup location \(s_i \in X\) and a drop-off location \(t_i \in X\). We use the terms \emph{package} and \(s\)-\(t\) pair interchangeably. Each drone can carry at most one package at a time.
 
One of the main objective functions in collaborative delivery problems is the total completion time, also
known as the sum of completion times. For a given solution, the completion time is computed for every
drone, and the objective function value is the sum of all these values. 
The task of our problem is to find a schedule for delivering all the packages from the start node to the destination node while minimizing the total completion time.   
 
We define a schedule $\pi_j$ for drone \( j \) as an ordered (by pickup time $t$) set of triples \( \{(u_i, w_i, t)\} \), where the drone picks up package \( P_i \) at location \( u_i \) at time \( t \) and delivers it to location \( w_i \). For any two consecutive triples in the schedule \( \{(u_i, w_i, t)\} \) and \( \{(u_{i'}, w_{i'}, t')\} \), the pickup time of the second must be at least the previous pickup time plus the travel time from \( u_i \) to \( w_i \) and from \( w_i \) to \( u_{i'} \). 
The time $t$ of the first triple \( \{(u_i, w_i, t)\} \) in $\pi_j$ must be at least the time it takes the drone to travel from its depot $r_j$ to $u_i$. 

The schedule for each package \( P_i \) can be obtained by collecting all triples involving package \( i \), i.e., \( \bar{\pi}_i= \{(u_i, w_i, t), ...\} \), from the union of all drone schedules \( \bigcup_j \pi_j \), and sorting them by the associated time \( t \). Note that for each package \( P_i \), the first triple should have \( u_i = s_i \), and the final triple should have \( w_i = t_i \); these may be the same triple.  
A package can only be picked up at a location \( u_i \) if it is the source \( s_i \), or after it has been dropped off at \( u_i \) by another drone (or the same drone) at an earlier time.

\begin{definition}[Min-Sum Collaborative  Delivery Problem (Min-Sum CD Problem)]
Given a metric space $(X, d)$, a set of $m$ packages $\mathcal{P} = \{s_i, t_i\}_{i=1}^m$ where $s_i$ is the pick-up location and $t_i$ is the drop-off location, and $k$ drones each with a depot $r_j$ and speed $p_j$ for each drone~$j$, 
the goal is to find schedules (or routes) \( \{\pi_j\}_{j \in [k]} \) such that the schedule for each drone \( j \) originates at its initial location \( r_j \), and each package \( i \) is transported from its source \( s_i \) to its destination \( t_i \).  
The objective is to minimize the total cost, defined as the sum of travel times across all drone routes: 
\[
\sum_{j=1}^k \left(t_{f_j} + \frac{d(u_{f_j}, w_{f_j})}{p_j}\right),
\]  
where \( (u_{f_j}, w_{f_j}, t_{f_j}) \) is the final triple in the
schedule $\pi_j$ of drone $j$. Note that drone $j$ finishes
its schedule at time $t_{f_j} + \frac{d(u_{f_j}, w_{f_j})}{p_j}$
because \( t_{f_j} \) is the time
at which it makes its final pick-up (at location $u_{f_j}$), and then
its final move is to take the package $P_{f_j}$ from $u_{f_j}$
to $w_{f_j}$ with speed $p_j$.
\end{definition}

We consider Min-Sum CD with both fine-grained and course-grained
collaboration. In the remainder of the paper, we refer to these
variants of the problem as the preemptive and the non-preemptive
variant, respectively.

\begin{definition}[Preemptive vs Non-preemptive]
Considering the output routes, we have the following two versions of Min-Sum CD:
\begin{itemize}
\item In the \emph{preemptive Min-Sum CD} problem, after
picking up a package from its source, it may (repeatedly) be left at an intermediate location
before being picked up by another (or the same) drone and transported further on its route to the destination. In this case, the schedule for each package may include multiple triples.

\item In the \emph{non-preemptive Min-Sum CD} problem, a package, once picked up from its source, is carried by the drone until it is dropped off at its destination. In this case, the schedule for each package includes only a single triple \( (s_i, t_i, t) \).
\end{itemize}
\end{definition}

In this paper, our focus will be on designing algorithms for the \emph{non-preemptive Min-Sum CD} problem, as we can demonstrate that the impact of preemption on the Min-Sum CD problem is limited: We show in the following section that the gap between the cost of the optimal preemptive and non-preemptive tours is no more than a constant factor.

\section{From Preemptive  to Non-preemptive}
\label{sec:FromNonpreToPre}

We show that the gap between the optimal collaborative preemptive and non-preemptive
tours is bounded.

\begin{theorem}
\label{thm:FromNonpreToPreMinSum}
An optimal schedule for the \emph{preemptive Min-Sum CD} problem with total cost $\OPT$ can be transformed into a non-preemptive schedule
with cost at most $3\cdot \OPT$.
\end{theorem}

\begin{proof}
Consider an optimal schedule $\Pi^*$ for the  \emph{Preemptive Min-Sum CD} problem, with $\pi^*_i$ denoting the schedule of drone~$i$ in the optimal schedule. 
 The cost of this optimal preemptive solution, $\OPT$, also serves as a lower bound for the optimal solution to the \emph{Non-preemptive Min-Sum CD} problem. Given an optimal schedule $\Pi^*= \{\pi^*_i\}_{i\in[k]}$, we construct a feasible schedule  for the non-preemptive Min-Sum CD problem with value at most $3\cdot \OPT$. 

The construction of a feasible non-preemptive schedule $\Pi$ from $\Pi^*$ is done via the following algorithm.  Recall that the drones are indexed by non-increasing speed.

	\begin{algorithm}[H]
 \caption{Reducing a preemptive solution to a non-preemptive solution}\label{alg:FromNonpreToPreMinSum}
 		\begin{algorithmic}[1]
 			\State $\Pi=\emptyset$, $i=1$.
 			\While{$i\le k$}
 				\State Obtain from $\pi^*_i$ a schedule $\pi_i$ serving all unscheduled packages in the order of serving in $\pi^*_i$.
 				\State Set $\Pi\gets \Pi\cup\{\pi_i\}$ and $i\gets i+1$.
 			\EndWhile
 		\end{algorithmic}
 		\end{algorithm}
 		
\begin{figure}
\centering
\begin{subfigure}{.43\textwidth}
  \centering
  \includegraphics[width=.8\linewidth]{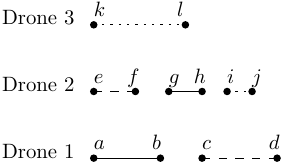}
  \caption{Original preemptive schedule}
\end{subfigure}%
\begin{subfigure}{.55\textwidth}
  \centering
  \includegraphics[width=1\linewidth]{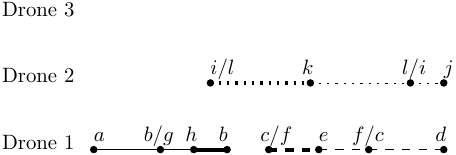}
  \caption{Resulting non-preemptive schedule}
\end{subfigure}
\caption{Consider the preemptive schedule with three drones in Figure (a), represented by solid, dotted, and dashed lines, corresponding to three different packages, respectively. Drones $2$ and $3$ are involved in parts of the dotted package schedule, while drones $1$ and $2$ are involved in parts of both the dashed and the solid package schedules. Note that points $b/g$, $c/f$, and $i/l$ are co-located. 
In the constructed non-preemptive schedule in Figure (b), additional travel distances are incurred for each package due to a drone deviating from its original route—either to pick up a package or to return to a point where it can resume its original schedule after a delivery. These additional distances are represented by thicker lines matching the style of the corresponding package in Figure (b). 
For the solid package, the fastest drone, drone $1$, first serves the sub-path $a \rightarrow b$, then the sub-path $g \rightarrow h$ to complete the delivery, and finally returns to point $b$ to resume the original schedule. 
For the dashed package, drone $1$ travels from point $c$ (which shares the location with $f$) to the pickup point $e$, then serves the sub-path $e \rightarrow f$ followed by the sub-path $c \rightarrow d$ to complete the delivery.
}
\label{fig:fromcollaborativetopreemptive}
\end{figure}

The schedules $\pi^*_i$ that constitute $\OPT$ are processed in increasing order of~$i$, i.e., the schedule of the fastest
drone is processed first and that of the slowest drone last.
We say a drone serves package $j$ if it carries $j$ at least once during its path.
Assume drone $i$ serves $k_i$ different packages $P_{i_1},\ldots,P_{i_{k_i}}$, indexed
in the order in which $i$ picks them up the first time, that aren't served by any
drone $i'$ with $i'<i$. For each $1\le j\le k_i$,
assume drone $i$ picks up $P_{i_j}$ at vertex $v_{i_j}$ and drops it off at $w_{i_j}$
the first time it carries $P_{i_j}$. For each $j$ with $1\le j\le k_i$, we
modify the tour of drone $i$ as follows: after reaching $v_{i_j}$, the drone
moves to $s_{i_j}$, picks up $P_{i_j}$ there, carries it directly to $t_{i_j}$,
then moves to $w_{i_j}$ and continues its original tour. After processing $\pi^*_i$
in this way to obtain a schedule $\pi_i$ for drone~$i$, the packages $P_{i_1},\ldots,P_{i_{k_i}}$
are served non-preemptively by drone~$i$. If $i<k$, we proceed to handle drone $i+1$
in the same way. 
Figure~\ref{fig:fromcollaborativetopreemptive} gives an example of preemptive tours and the constructed non-preemptive tours.

The algorithm transforms an optimal preemptive solution $\{\pi^*_i\}_{i\in[k]}$ into a feasible non-preemptive one $\{\pi_i\}_{i\in[k]}$.
For each package $P_j$, let $i(j)$ denote the fastest drone that serves $P_j$ in the preemptive schedule. Let $\bar{\pi}^*_j$ represent the set of triplets involving package $P_j$ that are traversed in the preemptive schedule, and let $w(\bar{\pi}^*_j)$ denote the total length of these triplets. 
Then $\sum_j (w(\bar{\pi}^*_j) / p_{i(j)})$ is a lower bound
on the total cost of the preemptive solution, as package $P_j$
is served only by drone $i(j)$ and slower drones. When transforming the solution
into a feasible non-preemptive one using the algorithm described above, 
the increase in the total cost is at most
$\sum_j (2 w(\bar{\pi}^*_j) / p_{i(j)})\le 2\OPT$.  
This is because letting drone  
$i(j)$ that originally serves only the subpath of $\bar{\pi}^*_j$ from some $v_j$ to some $w_j$ in the preemptive
solution now serve the whole package $P_j$ requires extra trips from
$v_j$ to $s_j$ and back, and from $w_j$ to $t_j$ and back. This results in an increase in total distance of at most $2w(\bar{\pi}^*_j)$, because the total
length of the subpaths from $v_j$ to $s_j$ and from $w_j$ to $t_j$ is
at most $w(\bar{\pi}^*_j)$ as these subpaths are disjoint parts of $\bar{\pi}^*_j$.
In conclusion, the constructed
non-preemptive solution  has total cost at most $3\cdot \OPT$ because
the original preemptive solution has cost at most $\OPT$ and the increase
in the total cost when transforming the solution into a non-preemptive
one is at most $2\cdot \OPT$.
\end{proof} 

Theorem~\ref{thm:FromNonpreToPreMinSum} implies the following.

\begin{corollary}
    An $\alpha$-approximation algorithm for the \emph{Non-preemptive Min-Sum CD} problem implies a
$3\alpha$-approximation algorithm for the \emph{Preemptive Min-Sum CD problem}.
\end{corollary}

The following lemma shows that the total length of all shortest package routes,
divided by the speed of the fastest drone, is a lower bound on the optimal cost of the Min-Sum CD problem.

\begin{lemma}
\label{lem:lengthoverpklowerbound}
Let an instance $I$ of Preemptive or Non-preemptive Min-Sum CD be given,
and let $p_1$ be the speed of the fastest drone.
Then $$\frac{\sum_{i=1}^m d(s_i,t_i)}{p_1}$$ is a lower bound
on the optimal sum of completion times for~$I$.
\end{lemma}

\begin{proof} 
Let $\Pi^*$ be an optimal solution to instance $I$.
For each $1 \le i \le m$, let $\bar{\pi}_i$ denote the set of triplets involving package $P_i$ in $\Pi^*$.  
In the non-preemptive case, $\bar{\pi}_i$ contains only a single triplet $(s_i, t_i, t)$.  
In the preemptive case, $\bar{\pi}_i$ may be a concatenation of triplets handled by multiple drones. 
The total cost of $\Pi^*$ is at least $\left( \sum_{i=1}^m w(\bar{\pi}_i) \right) / p_1$,  
where $w(\bar{\pi}_i)$ denotes the total length of the triplets in $\bar{\pi}_i$. 
Since $w(\bar{\pi}_i) \ge d(s_i, t_i)$ for each $i$, the lemma follows.
\end{proof}

\section{Non-preemptive Min-Sum CD}
\label{sec:Non-preemptiveMin-sumCD}
We first consider the special case where all drones are initially located
in the same depot and show that using only the fastest drone provides
a constant-factor approximation (Section~\ref{sec:samedepot}).
After that, we tackle the general case in Section~\ref{sec:arbdepot}.

\subsection{Identical Depot Locations}\label{sec:samedepot}
Assume that all drones are located at the same depot. We start with the following lemma.

 \begin{lemma}
 \label{lem:single-depot-MST}
      Let an instance $I$ of non-preemptive Min-Sum CD where all drones have the same
     depot location be given.
     Denote by $T$ a minimum spanning tree for the depot location and the sources
     of all packages, and use $w(T)$ to denote the sum of its edge lengths.
     Then $w(T)/p_1$ is a lower bound on the optimal sum of completion times,
     where $p_1$ is the speed of the fastest drone.
 \end{lemma}

 \begin{proof}
     Let $\pi^*_1,\pi^*_2,\ldots,\pi^*_k$ be the schedules (or paths) followed by the $k$ drones
     in an optimal solution. The union of these paths (after short-cutting all
     nodes that are not depot or packages sources) is a connected graph
     that contains the depot and the sources of all packages and uses only
     direct edges between such nodes (as opposed to, say, a Steiner tree that
     could use additional Steiner nodes to connect the nodes). As the minimum
     spanning tree has minimum total length among all such connected graphs,
     the total length of the edges of those paths is at least $w(T)$.
     Since no drone has speed larger than $p_1$, the total completion
     time of the optimal solution is at least $w(T)/p_1$.
 \end{proof}

 \begin{theorem}[Min-Sum CD with Identical Depot Locations]
\label{thm:Non-preemptiveMin-sumCDIdenticalDepot4}
 There is a $4$-approximation algorithm for the non-preemptive Min-Sum CD problem when all drones have the same depot location.
\end{theorem}

\begin{proof}
     The algorithm first computes a minimum spanning tree $T$ for the depot
    and all source locations. Then it doubles every edge of $T$
    and creates a Eulerian tour $\pi$ of $T$. Next, 
    it attaches to each package source $s_i$
    at its first occurrence in $\pi$ the edges $(s_i,t_i)$ (to deliver package $P_i$) and $(t_i,s_i)$ (to return to $s_i$), yielding the tour $\pi'$.
    Finally, using the fastest drone to follow tour $\pi'$ starting at the
    depot and ending when the last package is delivered is output as the solution.

    The sum of completion times of the solution is equal to the completion
    time of the fastest drone, as no other drone is used. The completion
    time of the fastest drone is at most $$\frac{d(\pi')}{p_1}
    \le \frac{2w(T)+2\sum_{i=1}^m d(s_i,t_i)}{p_1}\,,$$
    which is bounded by $2\OPT+2\OPT=4\OPT$ by Lemmas~\ref{lem:single-depot-MST} and~\ref{lem:lengthoverpklowerbound}.
\end{proof}

\subsection{Arbitrary Depot Locations}\label{sec:arbdepot}
 We translate the Min-Sum CD problem into a certain tree combination problem.
An instance of the tree combination problem consists of a set of trees, referred to
as \emph{original trees}, each associated with a drone speed~$p_i$.
We can add edges to connect different trees. A tree formed by a number
of original trees together with the edges that have been added to connect
them into a single tree is called a \emph{large tree}. For every large
tree~$T$, the cost of $T$ is defined to be the total edge length of $T$
divided by the fastest drone speed associated with any of the original
trees that are included in~$T$.

\begin{definition}[Tree Combination Problem]
\label{def:TreeCombination}
Given a metric space $(V, d)$ and a set of trees $\mathcal{T} = {T_1, T_2, \ldots, T_k}$, where $V(T_i) \subseteq V$ and each tree $T_i$ is rooted at a drone location $r_i \in V$ with a corresponding drone speed~$p_i$, the objective is to combine trees  so as to minimize the sum,
over all the resulting large trees, of the total edge length of the large tree  divided
by the highest speed among the drones associated with the original trees included in
the large tree.
\end{definition}

\begin{lemma}[Reducing Min-Sum CD to a Tree Combination Problem]
\label{lem:ReducingMin-sumCDTtoCombiantion}
Every given instance of the Min-Sum CD problem can be reduced
in polynomial time to
an instance of the Tree Combination Problem such that an
$\alpha$-approximation algorithm for the Tree Combination Problem
implies a $6\alpha$-approximation algorithm for the Min-Sum CD problem. 
\end{lemma}

\begin{proof}
Given an instance of the min-sum CD problem, we construct an instance of the tree combination problem with the following algorithm:

\begin{algorithm}[H]
 \caption{Constructing Trees $((V, d), \{s_i, t_i\}_{i\in [m]}, \{r_i\}_{i\in[k]})$}\label{alg:CreatCombiantion}
 		\begin{algorithmic}[1]
 			\State Create a new graph $G'$ from the metric $(\{s_i\}_{i\in[m]}\cup \{r_i\}_{i\in[k]}), d)$ by contracting the set of depots $R=\{r_1, r_2, \ldots, r_k\}$ into a single node $r$.
\State Compute the minimum spanning tree (MST) $T$ of the graph $G'$.
\State Obtain the forest $\{T_i\}_{i\in [k]}$ from $T$ by un-contracting the node $r$ back into the original depots in $R$, such that each tree $T_i$ in the forest is rooted at a distinct depot $r_i$.
\State For each package source vertex $s_i$, connect vertex $t_i$ to $s_i$ in its respective tree. 
 		\end{algorithmic}	\end{algorithm}

The algorithm begins by contracting all  depots into a single node. Specifically, this involves: (i) introducing a new node $r$ and removing the depots in $R$, (ii) for each depot $r_i \in R$, every edge $(s, r_i)$ (for every source~$s$) induces an edge $(s, r)$, and (iii) setting $d(s, r) = \min_{i\in [k]} d(s, r_i)$. The algorithm then computes a minimum spanning tree (MST) in the contracted graph. In Line~3, the MST is broken into a set $\{T_i\}_{i\in[k]}$ of $k$ disjoint trees by un-contracting the nodes in $R$. This un-contraction means that an edge $(s, r)$ is mapped back to an edge $(s, r_i)$, where $d(s, r) = d(s, r_i)$. Note that, by construction, every tree $T_i$ is rooted at $r_i$ and each vertex is covered in exactly one tree. Tree $T_i$ is associated with drone speed $p_i$.
If several drones are located at the same depot~$r_i$, the drone with the fastest speed among the
drones at that depot is chosen.

Note that a solution to the constructed instance of the tree combination problem with objective value $x$ can
be used to construct a solution to the non-preemptive Min-Sum CD problem with
sum of completion times at most $2x$ as follows: For each large tree $T'$ with
highest drone speed $p_i$, let the drone with speed $p_i$ traverse an Euler tour of $T'$.

\begin{observation}
\label{obs:constructtrees}
For any subset $U$ of the trees $\{T_i \mid 1\le i\le k\}$,
the total length of the optimal tours for the packages contained in trees of $U$
is at least one half of the total length of the trees in $U$.
\end{observation}

\begin{obsproof}
The total length of the trees in $U$ is equal to the weight of the MST
of the subgraph $G''$ of $G'$ induced by $r$ and all the sources in $U$, plus the sum of
$d(s_j,t_j)$ for all packages $P_j$ whose source $s_j$ is in $U$.
Let $\Pi^*$ be optimal tours for serving the packages whose sources
are in $U$. For each tour $\pi^*$ in $\Pi^*$, consider the \emph{shortened tour}
obtained from $\pi^*$ by only visiting the sources on $\pi^*$
and short-cutting all other nodes. The union of these shortened tours forms
a spanning subgraph of $G''$ and hence has total length at least the
weight of the MST of $G''$. Therefore, the weight of the MST is at most
the total length of the shortened optimal tours, and hence at most the
total length of the optimal tours. Furthermore, a separate path of length at least $d(s_j,t_j)$
must be included in the tours in $\Pi^*$ for each package with $s_j$ in $U$.
Hence, the total length of the optimal tours is at least the
sum of $d(s_j,t_j)$ over all packages $P_j$ with $s_j$ in $U$.
This shows that the total length of the trees in $U$ is at most twice
the total length of the optimal tours.
\end{obsproof}

Let $\Pi^*= \{\pi^*_1, \pi^*_2, \ldots, \pi^*_k\}$ be an optimal schedule, with total cost~$\OPT$, for the given instance of the min-sum CD problem. If the package paths of each constructed tree $T_i$ are entirely contained within one of the optimal tours, and if the drone used to execute an optimal tour is one of the drones
associated with a tree $T_i$ whose packages that tour serves, then the lemma can be proved as follows. Consider one
optimal tour $\pi^*_j$, and let $U_j$ denote the set of trees $T_i$ that are served
by $\pi^*_j$. We combine the trees in $U_j$ into a large tree $T'_j$, using
the edges of an MST in the graph $H$ on $U_j$ where every tree $T_i$ in $U_j$
is contracted to a single node.
The total length of the edges used for the combination is at most the length
of the tour $\pi^*_j$, as $\pi^*_j$ induces a connected subgraph of $H$. Hence,
by Observation~\ref{obs:constructtrees}, the total length of $T'_j$
is at most $3$ times the total length of $\pi^*_j$. As the same argument holds
for every optimal tour $\pi^*_j$, the constructed solution to the tree
combination problem has objective value at most $3\OPT$.
An $\alpha$-approximation algorithm for the tree combination problem
hence produces a solution with objective value at most $3\alpha\OPT$,
giving a solution with total cost at most $6\alpha\OPT$
for the non-preemptive Min-Sum CD problem.

In the general case,
we can apply the following combination procedure to obtain
a solution to the tree combination problem, for which we can show that the cost
is at most three times the cost of the optimal solution to the Min-Sum CD problem.
We say that a tour $\pi_j$ or a tree $T_j$ is \emph{trivial} if it serves no packages
and \emph{non-trivial} otherwise.
In other words, a tour $\pi_j$ or tree $T_j$ is trivial if the drone from that depot is not
used to serve any packages.

\begin{algorithm}[H]
 \caption{Construct Large Trees $(\{\pi^*_i\}_{i\in [k]}, \{T_i\}_{i\in [k]})$}\label{alg:ConstructLargeTrees}
 		\begin{algorithmic}[1]
\For{$j = 1$ to $k$}
    \If{$\pi^*_j$ is non-trivial or $T_j$ is non-trivial and has not yet been merged}
        \State set $\Pi = \{ \pi^*_j \}$ and $\mathcal{T}=\{ T_j \}$
        \While{the set of packages served in $\Pi$ and $\mathcal{T}$ are different}
            \If{there exists a package $l$ served by a tour in $\Pi$ that is not served
            by one of the trees in $\mathcal{T}$}
                \State let $T_i$ (necessarily $i> j$) be the tree that serves $l$
                \State insert $T_i$ into $\mathcal{T}$
                \State insert $\pi^*_i$ into $\Pi$
            \EndIf
            \If{there exists a package $l$ served by a tree in $\mathcal{T}$ that is not served
            by one of the tours in $\Pi$}
                \State let $\pi^*_i$ (necessarily $i> j$) be the tour that serves $l$
                \State insert $\pi^*_i$ into $\Pi$
                \State insert $T_i$ into $\mathcal{T}$
            \EndIf
        \EndWhile
        \State form a large tree $L_i$ from the trees in $\mathcal{T}$ (using an MST computation to find edges of minimum total length to be added) and add it to the solution
    \EndIf
\EndFor
 \end{algorithmic}
 \end{algorithm}

Assume Algorithm~\ref{alg:ConstructLargeTrees} constructs $\ell$ large trees.
Denote these large trees by $L_{i_1},L_{i_2},\ldots, L_{i_\ell}$
with $1\le i_1<i_2<\cdots<i_\ell \le k$, where any large tree $L_{i_q}$ consists
of tree $T_{i_q}$ and some number of trees $T_j$ for $j> i_q$.
Let $\mathcal{T}(L_{i_q})$ denote the set of these trees.
Note that $L_{i_q}$ serves the packages served by $\pi^*_{i_q}$ and
some number of tours $\pi^*_j$ for $j>i_q$.
Let $\Pi(L_{i_q})$ denote
the set of optimal tours whose packages $L_{i_q}$ serves.
Note that the sets $\mathcal{T}(T_{i_q})$ of trees form a partition
of the original set of trees, that the sets $\Pi(L_{i_q})$ form a partition
of the set of optimal tours, and that the set of indices of the trees in
$\mathcal{T}(T_{i_q})$ is the same as the set of indices of the optimal tours
in $\Pi(L_{i_q})$.
At the point Algorithm~\ref{alg:ConstructLargeTrees} forms
the large tree $L_{i_q}$, the variables $\mathcal{T}$ and $\Pi$
correspond to $\mathcal{T}(L_{i_q})$ and $\Pi(L_{i_q})$, respectively.

By Observation~\ref{obs:constructtrees}, the total length
of the trees $T_j$ in $\mathcal{T}(L_{i_q})$ is at most
twice the total length of the tours in $\Pi(L_{i_q})$.
Furthermore, the total length of the edges used to merge the $T_j$ in
$\mathcal{T}(L_{i_q})$ into the single large tree $L_{i_q}$ is
at most the total length of the tours in $\Pi(L_{i_q})$.
To see this, we will show the following claim: Whenever a tree $T_i$ and a tour
$\pi^*_i$ get
added to $\mathcal{T}$ and $\Pi$, respectively, there is a tour $\pi^*_g$ in $\Pi$
that contains a path from a node in $T_g\in \mathcal{T}$
to a node in~$T_i$. From this claim it then follows that the optimal tours in $\Pi(L_{i_q})$
contain paths that make the trees in $\mathcal{T}(L_{i_q})$ a connected
structure and hence have total length at least the total
length of the edges the algorithm chooses to connect those
trees into $L_{i_q}$ in Line~16.

To show the claim, first note that throughout the execution of Lines~3--15
it holds that every tree $T_i$ in $\mathcal{T}$ contains the
same depot $r_i$ as the corresponding optimal tour $\pi^*_i$.
Now, if $T_i$ and $\pi^*_i$ get added in Lines~7--8 because another tour
$\pi^*_g$ in $\Pi$ serves a package $l$ that is not yet served
by a tree in $\mathcal{T}$, then $\pi^*_g$ contains a path from $r_g$
to $s_l$, so $\pi^*_g$ contains a path connecting $T_g$ and $T_i$.
If $T_i$ and $\pi^*_i$ get added in Lines~12--13 because a tree
$T_h$ in $\mathcal{T}$ serves a package $l$ that is not yet served
by an optimal tour in $\Pi$ and that is served by $\pi^*_i$, then
$\pi^*_i$ contains a path from $r_i$ to $s_l$, which is a path
connecting $T_i$ and $T_h$. Thus, the claim follows.

In the remainder of the proof, we use $\len(E)$ to denote the total
length of $E$, where $E$ can be a tree, a set of trees, a tour, or
a set of tours. Furthermore, we define $\cost(E)=\frac{\len(E)}{p_g}$,
where $p_g$ is the speed of the drone that serves the requests associated
with $E$ and whose identity will be clear from the context.
From the discussion above, we obtain
that the total length $\len(L_{i_q})$ of the large tree $L_{i_q}$
is at most $3$ times the total length $\len(\Pi(L_{i_q}))$ of the tours in $\Pi(L_{i_q})$: The trees in $\mathcal{T}(L_{i_q})$ have total length at most
$2\cdot \len(\Pi(L_{i_q}))$, and the edges used to connect them into
a large tree have total length at most $\len(\Pi(L_{i_q}))$.

Next, observe that the large tree $L_{i_q}$ is served by
drone $i_q$, and that this is the fastest of the drones associated
with any tree in $\mathcal{T}(L_{i_q})$, and thus also the fastest of the drones
associated with any of the optimal tours in $\Pi(L_{i_q})$.
We thus have
\begin{align*}
\cost(L_{i_q})& = \frac{\len(L_{i_q})}{p_{i_q}}
\le \frac{3\cdot \len(\Pi(L_{i_q}))}{p_{i_q}}
= \sum_{\pi^*_g \in \Pi(L_{i_q})} \frac{3\cdot \len(\pi^*_g)}{p_{i_q}}\\
& \le \sum_{\pi^*_g \in \Pi(L_{i_q})} \frac{3\cdot \len(\pi^*_g)}{p_{g}}
= 3 \sum_{\pi^*_g \in \Pi(L_{i_q})}\cost(\pi^*_g)
\end{align*}
where the last inequality follows from $p_g\le p_{i_q}$ for
all $\pi^*_g\in \Pi(L_{i_q})$ because $g\ge i_q$.
As the sets $\Pi(L_{i_q})$ form a partition of $\{\pi^*_1,\pi^*_2,\ldots,\pi^*_k\}$,
this implies:
$$
\sum_{j=1}^\ell \cost(L_{i_j}) \le 3 \sum_{j=1}^\ell \sum_{\pi^*_g \in \Pi(L_{i_q})} \cost(\pi^*_g)
= 3 \sum_{g=1}^k \cost(\pi^*_g)
$$
Therefore, the total cost of the large trees
is at most 3 times the total cost of the optimal tours.
Hence, an $\alpha$-approximation algorithm for the
tree combination problem yields a tree combination solution
with total cost at most $3\alpha\cdot \OPT$, which then gives
a solution to the Min-Sum CD problem with cost at most
$6\alpha\cdot\OPT$ as described earlier.
\end{proof}

Our problem can be reduced to solving the Tree Combination problem, where the input trees are the outputs of Algorithm~\ref{alg:CreatCombiantion}. Without loss of generality, the input consists of \( k \) trees, denoted as \( T_1, T_2, \dots, T_k \), where \( w(T_i) \) represents the total length of the tree \( T_i \). Each tree $T_i$ is rooted at a drone location \( r_i \) with an associated drone speed \( p_i \).

We introduce some notation for this problem as follows:   
Let $T = (V, E)$ represent one of the combined trees, where $V$ is a subset of the input trees $\mathcal{T}$, and $E$ consists of the shortest edges between the corresponding original trees. 
For an edge $e = (u, v) \in E$ connecting two trees $u, v$, the edge length $\ell(e)$ in the combined tree $T$ is defined as \[\ell(e) = \min_{a \in V(u), b \in V(v)} d(a, b)\]  
where $V(u)$ and $V(v)$ denote the sets of vertices in the original input trees corresponding to $u$ and $v$ in the combined tree $T$, respectively.  
If a tree $T$ contains more than two vertices (original trees), we assume it is rooted at a vertex with the highest speed. For every vertex $v \in V$, we use $V_v\subseteq V$ to denote the set of descendants of $v$, and $v$ itself. We use $E_v \subseteq E$ to denote all edges that include at least one vertex from the set $V_v$. 

\begin{observation}
\label{obs:OptimalCombinationTree}
Consider an optimal solution to the Tree Combination Problem with two speeds, $p_1 > p_2$. The following conditions are satisfied: 
\begin{enumerate}
    \item Each combined tree (with more than one original tree) contains exactly one tree with high speed (i.e., with speed $p_1$).
    \item For a combined tree $T = (V, E)$ and a vertex $v \in V$ different from the root, the total length of edges in $E_v$ is no more than $p_1/p_2 - 1$ times the total length of all vertices in $V_v$, i.e., $\sum_{e\in E_v} \ell(e) \le \sum_{u \in V_v} w(u)\cdot (p_1/p_2 - 1)$, where $w(v)$ is the total length of the corresponding original tree for vertex~$v$.
\end{enumerate}
\end{observation}

\begin{proof}
The first part of the observation is straightforward: If a combined tree does not contain a tree with the high speed, there is no need to connect these trees, as the total cost would only increase. On the other hand, if there are at least two trees rooted at drones with high speed within a combined tree, we can partition the combined tree into two separate trees by removing an edge connecting two subtrees that both contain a high speed tree. The resulting separated trees will have a lower total cost than the original combined tree since both separated trees contain a high speed tree and the total length is smaller. This contradicts the assumption that the combined tree is optimal in the solution.

 To demonstrate the second part of the observation, assume there is a vertex $v$ in the combined tree such that the total length of edges in $E_v$ exceeds $p_1/p_2 - 1$ times the total cost of all vertices in $V_v$, i.e., 
\[\sum_{u\in V_v} w(u) \cdot (p_1/p_2 -1)  <  \sum_{e\in E_v} \ell(e) \] 
\[\sum_{u\in V_v} w(u) /p_2  < \sum_{u \in V_v } w(u) /p_1 + \sum_{e\in E_v} \ell(e) /p_1 \]  
In this case, we can partition the combined tree into $1 + |V_v|$ trees (one tree $T'$ with high speed and $|V_v|$ trees with slow speed) by removing all edges in $E_v$.  
The resulting trees will have a lower total cost than the original combined tree, 
\[\frac{w(T')}{p_1} + \sum_{u\in V_v} w(u) /p_2 + w(v)/p_1 < \frac{w(T')}{p_1} + \sum_{u \in V_v} w(u) /p_1 + \sum_{e\in E_v} \ell(e) /p_1, \] 
which contradicts the assumption that the combined tree is optimal in the solution. 
\end{proof}

To solve the Tree Combination Problem with two speeds, we consider a Prize Collecting Steiner Tree Problem:  Given an undirected graph $G=(V, E)$, edge cost $c_e\ge 0$ for all edges $e\in E$, a set of selected root vertices $R\subset V$, and penalty $\pi_i\ge 0$ for all $i\in V$. The goal is to find a forest $F$ such that each tree in $F$ contains a root vertex $r\in R$, with the objective of minimizing the cost of the edges in $F$ plus the penalties of the vertices not in $F$, i.e., $\sum_{e\in F} c_e+\sum_{r\in V-V(F)} \pi_i$, where $V(F) $ is the set of vertices in the forest~$F$. 
The Prize Collecting Steiner Tree problem is normally defined for a single root,
but our formulation with a set $R$ of roots is equivalent as the vertices in $R$ can be
merged into a single root before an algorithm for the Prize Collecting Steiner
Tree problem is applied.

\begin{lemma}[Reducing the Tree Combination Problem with Two Speeds to a Prize Collecting Steiner Tree  Problem]
\label{lem:ReducingTreeCombinationtoPrizeCollecting}
An $\alpha$-approximation algorithm for the Prize Collecting Steiner Tree Problem implies an $\alpha$-approximation algorithm for the Tree Combination problem with two speeds. 
\end{lemma}
\begin{proof}
To formulate the Tree Combination Problem with Two Speeds as a Prize Collecting
Steiner Tree Problem, all trees are represented as vertices in the graph $G = (V, E)$, where each vertex corresponds to a tree. The edge between any two trees represents the cost of connecting the nodes corresponding to those trees. 
Specifically, the cost of an edge is given by the minimum distance between the nodes in the two trees, divided by the higher speed.
We designate all trees with high speed as the root vertices $R \subseteq V$, and all trees with slow speed are assigned penalties $\pi_i \geq 0$ for all $i \in V$. 
Specifically, the penalty for each vertex is given by the difference in the inverses of the speeds multiplied by the total length of the tree, which represents the cost incurred if the vertex is not connected to a high-speed tree.
The goal is to find a forest $F$ that contains multiple trees, each rooted at a node corresponding to a tree with high speed. The objective is to minimize the sum of the costs of the edges in $F$ and the penalties of the vertices not included in the forest $F$. 

The original Tree Combination Problem aims to minimize both the connection cost and the serving cost for each tree. In contrast, the constructed instance of the Prize Collecting Steiner Tree Problem focuses on the connection cost and only a partial serving cost for each tree. Specifically, for trees with high speed and slow-speed trees connected to a high-speed tree, the serving cost is not considered. For slow-speed trees not connected to a high-speed tree, only the difference in serving costs between high-speed and low-speed service is counted. The disregarded cost, which is the serving cost for all trees when served with high-speed, is incurred by any solution for the Tree Combination Problem. The total cost of a solution to the Tree Combination Problem is this disregarded cost, whose value is a fixed amount~$C$, plus the connection cost and penalties. For the optimal solution to the Tree
Combination Problem, the cost can be divided into the fixed cost $C$ and the cost of an optimal solution to the Prize Collecting Steiner Tree problem. Therefore, an $\alpha$-approximation algorithm for the Prize Collecting Steiner Tree Problem provides an $\alpha$-approximation algorithm for the original Tree Combination Problem. 
\end{proof} 

We use the following known result for the Prize Collecting Steiner Tree Problem

\begin{lemma}[Archer et al.~\cite{archer2011improved}]
\label{lemma:PrizeCollectingProblem}
There exists a $1.9672$-approximation algorithm for the Prize Collecting Steiner Tree Problem. 
\end{lemma}

By combining Lemmas~\ref{lem:ReducingMin-sumCDTtoCombiantion},   \ref{lem:ReducingTreeCombinationtoPrizeCollecting} and~\ref{lemma:PrizeCollectingProblem}, we obtain the following theorem. 

\begin{theorem}[Min-Sum CD with Two Speeds]
\label{thm:Non-preemptiveMin-sumCDTwoSpeeds}
There exists an $11.8$-approximation algorithm for the non-preemptive Min-Sum CD problem when the drones have only two distinct speeds. 
\end{theorem}

\subsection{Tree Combination Problem with $h\ge 3$ Speeds}

When we have drones with more than two speeds, the Prize Collecting Steiner Tree problem cannot be directly applied since the penalty is not fixed before we obtain the final solution.
Inspired by the study on the Heterogeneous Traveling Salesman Problem \cite{rathinam2020primal},
we employ a multi-level primal-dual algorithm to address this issue.
If vehicles travel at different speeds and the travel cost is defined as the ratio of the traveled distance to the vehicle's speed, then the travel cost adheres to the principle of monotone costs,
which is a prerequisite for the approach to apply.

In the Tree Combination Problem involving multiple-speed drones, we model the solution to include penalties for trees not visited by the highest-speed vehicle. These penalties are assessed at multiple levels, each calculated as the ratio of the distance to the speed difference between consecutive levels.

\subsubsection{LP Relaxation and Its Dual}
Consider the tree combination problem involving $k$ trees. Later, we will also use vertices to represent each tree. Suppose the $k$ vertices (trees) are ordered according to their respective speeds, and we categorize these vertices into $h \leq k$ distinct levels based on speed. The levels are organized from the fastest to the slowest, with level $1$ being the fastest and level $h$ being the slowest. Each level, $l\in [h]$, contains $h_l$ vertices. Formally, let $v_{l, i}$ represent the $i$-th vertex at speed level $l$. The set of vertices at level $l$ is denoted as $V_l = \{v_{l, i}\}_{i \in [h_l]}$. For each vertex $v_{l, i}$, where $l \in [h]$ and $i \in [h_l]$, the weight is defined as $w(v_{l, i})$, the subset of vertices it owns (in the corresponding tree)  is $V(v_{l, i})$,
 and its speed is denoted as $p_l$.

The input to the tree combination problem is a graph with vertex set \( V = \bigcup_{l} V_l \), edge set $
E = \left\{ (u, v) \mid u, v \in \bigcup_{l \in [h]} V_l \right\},$  where the length of an edge between two nodes \( v_{l, i}, v_{l', i'}\) is given by $
\ell(v_{l, i}, v_{l', i'}) = \min_{a \in V(v_{l', i'}),\ b \in V(v_{l, i})} d(a, b).$ 
A resulting tree rooted at $v_{l, i}$ is defined by $\tilde{T}(l, i)=(\tilde{V}_{l, i}, \tilde{E}_{l, i})$ with $ \tilde{V}_{l, i} \subseteq  V$ and $\tilde{E}_{l, i}\subseteq E$. The total cost of this tree is equal to the sum of the edge cost $\sum_{(u, v) \in \tilde{E}_{l, i}} \ell(u, v)/p_l$ and the vertex cost $\sum_{v \in \tilde{V}_{l, i}} w(v)/p_l$ if $ |\tilde{V}_{l, i}|\ge 1$, and is equal to $0$ otherwise.
(We set $\tilde{V}_{l, i}=\tilde{E}_{l, i}=\emptyset$ if none of the
resulting trees is rooted at $v_{l,i}$.) The objective is to identify a set of trees such that each vertex is included in exactly one of the resulting trees, while minimizing the total cost, i.e., $$\sum_{l \in [h], i\in [h_l]} \left( \sum_{(u, v) \in \tilde{E}_{l, i}} \ell(u, v)/p_l + \sum_{v \in \tilde{V}_{l, i}} w(v)/p_l \right).$$ 
Refer to Figure~\ref{fig:feasible_solution_instance} for an illustration of a tree combination problem involving nine trees across three levels. 

\begin{figure}[h] 
\centering
\includegraphics[width=\textwidth]{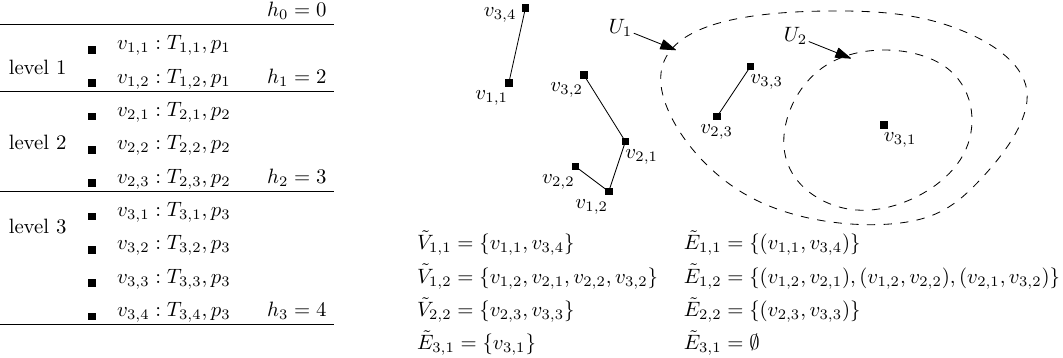}
\caption{An instance of the tree combination problem with nine vertices, each associated with one of three different speeds \( p_1 > p_2 > p_3 \). An illustration of a feasible solution shows four resulting trees.}
\label{fig:feasible_solution_instance}
\end{figure}

We design an approximation algorithm based on the primal-dual method.    
To formulate the problem as an integer linear program, we define two sets of integer variables with respect to the following vertex and edge sets:

\begin{itemize}
  \item The vertex set
  \[
  \bar{V}_l = \bigcup_{l' > l} V_{l'},
  \]
  which consists of vertices that may appear in the resulting trees rooted at a vertex in  level~\(l\), but are not initially in level~\(l\).

  \item The edge set
  \[
  \bar{E}_l = \left\{ e \in E \mid \text{both endpoints of } e \text{ have level } \geq l \right\},
  \]
  which consists of edges that may appear in the resulting trees rooted at  a vertex in level~\(l\).
\end{itemize}
The two sets of variables are:
\begin{itemize}
  \item 
  The first set of variables, denoted by \( x^l_e \), indicates whether an edge \( e \in \bar{E}_l \) is included in the resulting tree rooted at a vertex in \( V_l \), where \( l \in \{1, \ldots, h-1\} \).
  
  \item 
  The second set of variables, denoted by \( z_U^l \), is defined for \( l = 1, \ldots, h-1 \) and for any subset \( U \subseteq \bar{V}_l \). Specifically, \( z_U^l = 1 \) if \( U \) represents
  the subset of the
vertices with speed $p_{l+1}, p_{l+2}, ..., p_{h}$ that do not
get connected to a vertex with speed $p_l$ or faster.
\end{itemize}

Let the cost of traversing an edge $e \in \bar{E}_l$ by its root vertex (drone) with speed $p_l$ be denoted by 
$\cost_l(e) = \ell(e)/p_l$. 
Denote the cost of serving a vertex $v$ by its root vertex (drone) with speed $p_l$ as 
$w_l(v) = w(v)/p_l$. 
Let $\pi_l(U) = \sum_{v \in U} (w_{l+1}(v) - w_l(v))$ 
be the level-$l$ penalty if all vertices in $U \subseteq \bar{V}_l$ are served by a root vertex with speed less than or equal to $p_{l+1}$. 
Notice that if a vertex $v$ is served by a root vertex with speed $p_l$, then it will be included in the sets $U_1, U_2, \ldots, U_{l-1}$, and the total penalty incurred for this vertex will be 
$\pi(\{v\}) = \sum_{i=1}^{l-1} \pi_i(\{v\}) = w_l(v) - w_1(v)$. 
Similar to the tree combination problem with two distinct speeds, we disregard the (positive and fixed) cost of serving all vertices with the highest-speed drone in the integer program, as this does not affect the approximation ratio.

The Integer Programming (IP) formulation and its Linear Programming (LP) relaxation for the
problem that we are addressing are as follows:

\begin{align}
    \text{Min}   \quad &\sum_{l\in[h-1]}\sum_{e\in \bar{E}_l} \cost_l(e) x^l_e +\sum_{l\in[h-1]} \sum_{U\subseteq \bar{V}_{l}} \pi_l(U) z^l_U   \label{eq:minobjective} \\
    \text{st.} \quad & \sum_{e \in \delta_l(S)} x^l_e  \geq  \sum_{U: S \subseteq U \subseteq \bar{V}_{l-1}} z^{l-1}_U - \sum_{U: S \subseteq U \subseteq \bar{V}_l}  z^{l}_U \quad \forall S \subseteq \bar{V}_l, l\in[h-1], \label{eq:minconstraint1} \\
                            & z^0_V = 1; \quad z^0_U = 0 \quad \forall U \neq V
                            ,  \label{eq:minconstraint2} \\
                            & x^l_e \in\{0, 1\} \quad \forall e \in \bar{E}_l,  l\in[h-1],       \label{eq:minnonnegativity_x} \\
                            & z^l_U \in\{0, 1\}  \quad \forall U \subseteq \bar{V}_{l},   l\in[h-1].        \label{eq:minnonnegativity_y}
\end{align} 
Here in constraint~(\ref{eq:minconstraint1}), we let \( \delta_l(S) \) denote the set of edges in $\bar{E}_l$ with exactly one endpoint in \( S \), i.e.,  $
\delta_l(S) = \delta(S) \cap \bar{E}_l.$ 
If we consider solutions where, for each $l$, $z^l_U=1$ for exactly one $U$ and
$z^l_U=0$ for all other $U$, then
both terms $\sum_{U: S \subseteq U \subseteq \bar{V}_{l-1}} z^{l-1}_U$ and $\sum_{U: S \subseteq U \subseteq \bar{V}_l} z^l_U$ in constraint~(\ref{eq:minconstraint1}) only take binary values, either $0$ or $1$.
Then this constraint is
trivially satisfied except for the case when  $\sum_{U: S \subseteq U \subseteq \bar{V}_{l-1}} z^{l-1}_U=1$ and $\sum_{U: S \subseteq U \subseteq \bar{V}_l} z^{l}_U=0$. But
this case can occur only if there is at least one
vertex in $S$ that must be connected to a tree rooted  at level $l$. The constraint reduces
to  $\sum_{e \in \delta_l(S)} x^l_e\ge 1$ in that case and it is ensured that each such vertex is indeed connected 
to a tree rooted at level~$l$.
The IP formulation also allows solutions in which, for some~$l$, $z^l_U=1$ for more
than one set $U$, but this cannot increase the objective value of an optimal
solution.

The objective is to find a set of resulting trees such that each vertex is included exactly once, and the total cost — including both connection (edge) costs and vertex service costs — is minimized. 
If a vertex $v$ is connected to a tree rooted at level $l$, then its cost (penalty) is accounted for in every preceding level from $1$ to $l-1$. These costs are included in the penalties associated with the sets $U_1, U_2, \ldots, U_{l-1}$. Specifically, as already
mentioned earlier, the cumulative difference in cost for vertex $v$ from level $l$ down to level $1$ is given by:
\[
w_l(v)-  w_{1}(v)  = \sum_{j \in [l-1]} \left(  w_{l+1}(v)  -  w_l(v)  \right) = \sum_{j \in [l-1]} \pi_j(\{v\}).
\]

If we replace the integrality conditions $ x^l_e \in\{0, 1\} $ and   $z^l_U \in\{0, 1\}$ with $ x^l_e \ge 0$ and $ z^l_U\ge 0$ and take the dual of the linear program, we obtain the following:    

\begin{align}
    \text{Max}   \quad &  \sum_{S\subseteq \bar{V}_1} Y_1(S)  \label{eq:maxobjective} \\
    \text{st.} \quad & \sum_{S\subseteq \bar{V}_l: e \in \delta_l(S)} Y_l(S)  \le \cost_l(e), \quad e\in \bar{E}_l,   \forall l\in[h-1], \label{eq:maxconstraint1} \\
                            & \sum_{S: S \subseteq U} Y_l(S) - \sum_{S: S \subseteq U- V_{l+1}} Y_{l+1}(S) \le \pi_l(U) \quad \forall U\subseteq \bar{V}_{l},  \forall l\in[h-1],  \label{eq:maxconstraint2}
                         \\
                         &   Y_l(S) \geq 0 \quad \forall S \subseteq \bar{V}_l,  l\in[h-1].       \label{eq:maxnonnegativity_y}
\end{align}
The constraints specified in Equation~(\ref{eq:maxconstraint1}) pertain to edges, whereas those in Equation~(\ref{eq:maxconstraint2}) are related to sets of vertices. These constraints limit the dual values assigned to any subset of the set 
$U$ at each level, ensuring that the combined sum of accumulated values from a higher cost level does not exceed the benefits derived from serving them at the lower cost level.

\subsubsection{Primal-Dual Algorithm}
Before presenting the algorithm, we first introduce some notation regarding the constraints. 

\begin{definition}[Remaining Potential]
For a given set of vertices $U\subseteq \bar{V}_l$ and a dual feasible solution $Y_l(U)$, we say that $P(Y_l(U), \pi) = \pi_l(U) - (\sum_{S: S \subseteq U} Y_l(S) - \sum_{S: S \subseteq U- Y_{l+1}} Y_{l+1}(S))$ is the remaining potential of the set $U$ in $l$-th level. 
\end{definition}

The algorithm maintains a separate forest for each level: The forest of a level consists of trees rooted at vertices within that level to which vertices from higher levels are connected, and possibly some trees
that consist only of vertices from higher levels.
This organization is achieved through the application of the primal-dual moat-growing procedure, which is implemented independently and simultaneously across each level's graphs.
We define $F_l(t)$ to be the resulting forest at level $l \in [h]$ in the graph  $(V_l\cup \bar{V}_l, \bar{E}_l)$ by the end of iteration $t$ of the main loop.
For any two levels $l $ and $l'$, where $l < l'$, and at any iteration $t$, consider a component $C_i$ in forest $F_l(t)$ and  $C_j$ in forest $F_{l'}(t)$. Then $C_i$ is considered an ancestor of $C_j$, and $C_j$
 a descendant of $C_i$, if $C_i\supseteq C_j$.

A component can be classified as either \textbf{active}, \textbf{inactive}, or \textbf{frozen}. A component is considered \textbf{active} at the start of an iteration if it consists of higher-level vertices and its dual variable is allowed to increase during the iteration without violating any constraints in the dual problem. A component is deemed \textbf{inactive} if it contains a vertex at the corresponding level or if it is a descendant of an inactive component. 
A component is \textbf{frozen} if it has ceased growing due to constraint~(\ref{eq:maxconstraint2}), at which point its remaining potential becomes $0$.
If a component is inactive or frozen at the start of an iteration, its dual value will not change during that iteration.

When a component becomes frozen, meaning its remaining potential is reduced to zero, it implies that each of its descendants at the next higher level has either already become frozen or has become inactive.   

 While the dual LP has only variables $Y_l(S)$ for
 $S\subseteq \bar{V}_l$, $l\in [h-1]$, in our description of
 the primal-dual algorithm we consider variables
 $Y_l(S)$ for $S\subseteq V_l\cup \bar{V}_l$, $l\in [h]$ for
 ease of presentation. The values of the additionally introduced
 variables are zero at all times.

\paragraph{Initialization} 
 Initially, each forest $F_l(0)$ for level $1\le l\le h-1$ at time $0$ consists of singleton components, each of which is a vertex of $V_l\cup \bar{V}_l$. The singleton components containing only vertices from higher levels $\bar{V}_l$ are considered active, while any component containing a vertex with the speed $p_{l}$ is deemed inactive.
 The dual variables corresponding to all active and inactive components are initialized to zero, for every $l\in[h]$ and $ v \in V_l$, we have $Y_{l'} (\{v\}) = 0$ for all $ l'\in [l]$.

\paragraph{Main Loop} 
In each iteration of the primal-dual algorithm, the dual variables of all active components in all forests are simultaneously increased by the same amount as much as possible until at least one of the constraints in the dual problem becomes tight. If multiple constraints in both~(\ref{eq:maxconstraint1}) and~(\ref{eq:maxconstraint2}) become tight simultaneously during iteration $t$, only one constraint, either from~(\ref{eq:maxconstraint1}) or~(\ref{eq:maxconstraint2}), is chosen and processed based on the following procedure:

If constraints in~(\ref{eq:maxconstraint1}) become tight, then a tight constraint corresponding to the vertex with the lowest level (denoted as $l$) is chosen first. The edge corresponding to the chosen constraint is added to the forest, merging two components (denoted as $C_1$ and $C_2$). If both $C_1$ and $C_2$ do not contain a vertex with speed $p_l$, then the merged component remains active. However, if one of the components contains a vertex with speed $p_l$ (say vertex $v_{l, j} \in C_2$), then the merged component and all its active descendants become
inactive. 

If a constraint in~(\ref{eq:maxconstraint2}) becomes tight (i.e., it risks being violated if the current dual variables continue to grow) for some component, i.e., $P(Y_l(U), \pi)=0$, then the corresponding component $U$ in level $l$ is deactivated and becomes frozen.  

The main loop of the algorithm terminates when every component in all forests is either inactive or frozen.

\begin{figure}[htbp]
     \centering
     \begin{subfigure}[b]{0.48\textwidth}
         \centering
         \includegraphics[width=\textwidth]{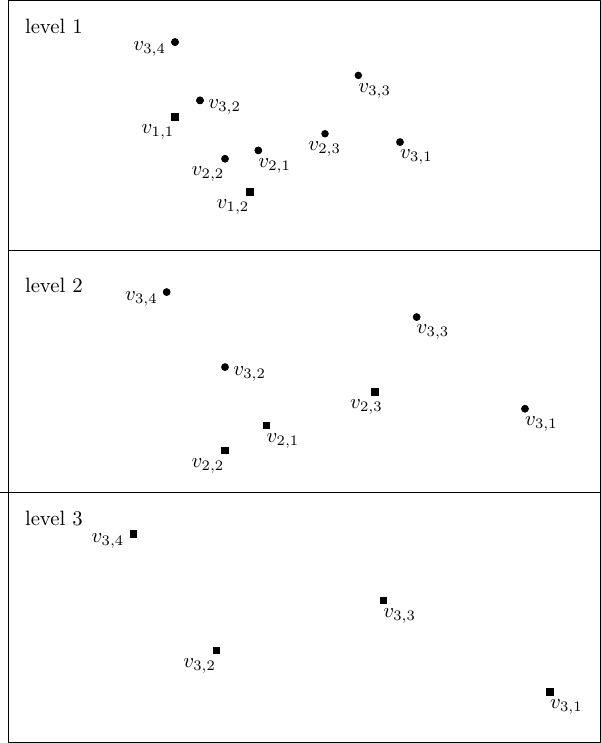} 
    \caption{Primal-dual: Initialization\\ For each level $l \in [3]$, the vertices $V_l$ are inactive and represented by squares, while the other vertices $\bar{V}_l$ are active and represented by disks. In each iteration, all active components' dual variables are increased equally until a dual constraint (either edge constraint or potential constraint) becomes tight.}
     \end{subfigure}
     \hfill
     \begin{subfigure}[b]{0.48\textwidth}
         \centering
         \includegraphics[width=\textwidth]{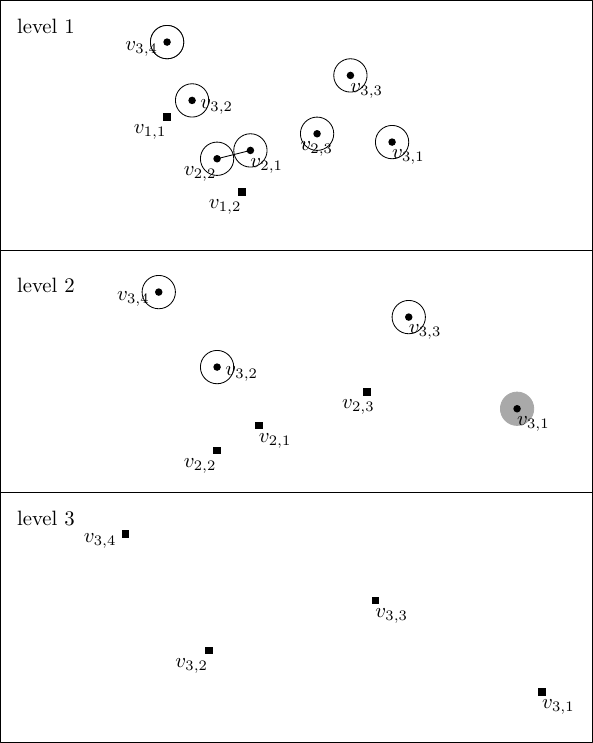} 
         \caption{Example of two tight constraints:\\
         Edge constraint: \( Y_1(v_{2,1}) + Y_1(v_{2,2}) = \text{cost}_1(\{v_{2,1}, v_{2,2}\}) \), thus 
         edge \( \{v_{2,1}, v_{2,2}\} \) is added to \( E_1 \), the components $\{v_{2,1}\}$ and $\{v_{2,2}\}$ in level $1$ become inactive and the merged component \( \{v_{2,1}, v_{2,2}\} \) becomes active.
         The component \( v_{3,1} \) in level $2$ becomes frozen.}

     \end{subfigure}
     \hfill
     \begin{subfigure}[b]{0.48\textwidth}
         \centering
         \includegraphics[width=\textwidth]{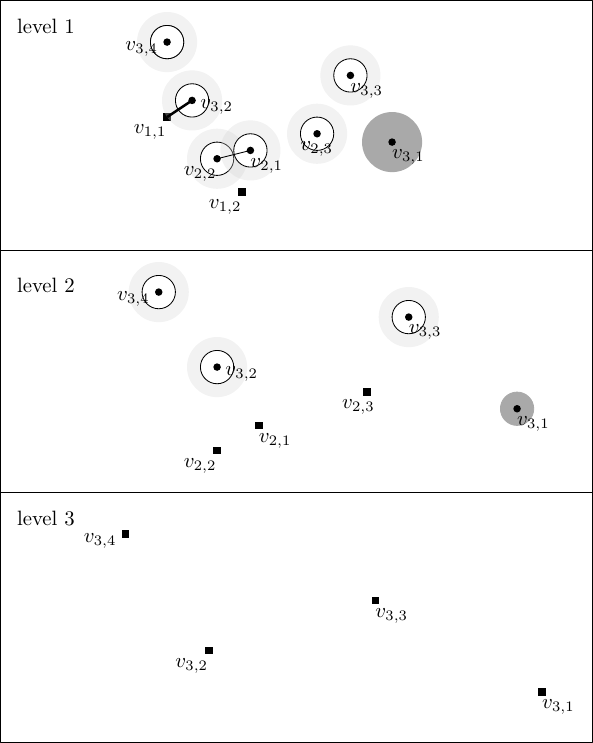} 
          \caption{ 
Edge \( \{v_{3,2}, v_{1,1}\} \) is added to \( E_1 \),
component \( v_{3,1} \) in level $1$ becomes frozen.}
     \end{subfigure}
      \hfill
     \begin{subfigure}[b]{0.48\textwidth}
         \centering
         \includegraphics[width=\textwidth]{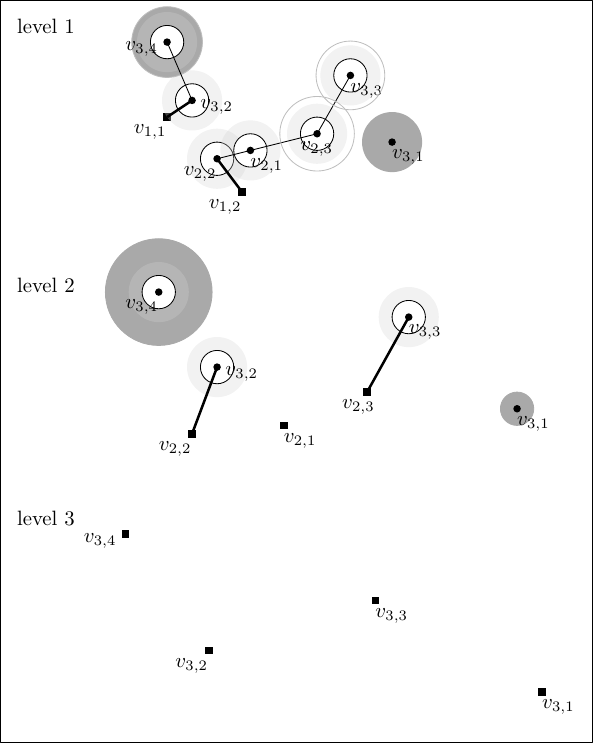}
          \caption{Final iteration. 
All components become inactive or frozen.}
     \end{subfigure}
        \caption{Illustration of the primal-dual algorithm}
        \label{fig:Illustrations_primal_dual}
\end{figure}

\begin{lemma}
\label{lem:primal-dual_runningtime}
 The main loop of the primal-dual algorithm terminates in $\sum_{l\in [h-1]} |V_l|\cdot (2l-1)$ iterations.
\end{lemma}
\begin{proof}
At the start of the main loop, the number of components in all the forest is $\sum_{l\in [h-1]} |V_l|\cdot l$ and the number of active component is $\sum_{l\in [h-1]} |V_l|\cdot(l-1) $. During each iteration of the main loop, the sum of the number of components and the number of active components in all
the forests decreases by at least $1$ and thus the main loop will require at most $\sum_{l\in [h-1]} |V_l|\cdot(2l-1)$ iterations. 
\end{proof}

Let the main loop of the primal-dual procedure terminate after $t_f$ iterations.    
Then, starting from level $l=1$, the pruning procedure selects a sub-forest $\tilde{F}_l$ of the level $l$ forest $F_l$ such that each tree of $\tilde{F}_l$ is rooted at one of the vertices in $V_l$ and each vertex is connected in exactly one of the resulting sub-forests.   
Before describing the pruning procedure, we need to define a few terms. The degree of a subgraph $C$ within a graph (tree) $T$ is defined as:
\[
|\{(u, v) : u \in C, v \notin C, (u, v) \in T\}|.
\]
A subgraph $C$ of a tree $T$ in level $l$ is referred to as a \textit{frozen subgraph} if $C$ is a component that was frozen in level $l$ during some iteration $t \le t_f$. In other words, its remaining potential $P(Y_l(C), \pi) =0$. A subgraph $C\subseteq F_l$ of a tree $T$ in level $l$ is called a \textit{pendent-frozen subgraph} if $C$ is a frozen subgraph and its degree is equal to $1$. A \textit{maximal pendent-frozen subgraph} is a pendent-frozen subgraph $C \subseteq T$ such that there is no other pendent-frozen subgraph $C_0 \subseteq T$ with $C_0 \supseteq C$. 
For each vertex $u\in V$, let $ C_{l}(t, u)$ denote the
component containing $u$ in level $l$ at the start of iteration $t$ of 
the main loop.
Define $\text{active}_{l}(t, u)$ as an indicator for whether $C_{l}(t, u)$ is active:
\[
\text{active}_{l}(t, u) = 
\begin{cases} 
1 & \text{if } C_{l}(t, u) \text{ is active}, \\
0 & \text{otherwise}.
\end{cases}
\]

We revisit some observations derived from the execution of the main loop. 
\begin{observation}[Rathinam et al.~\cite{rathinam2020primal}]
\label{obs:primal-dual_merge}
Some observations regarding the primal-dual algorithm:
\begin{enumerate}
\item Every component in the Primal-dual procedure is a tree graph. 
\item Considering edge constraints~(\ref{eq:maxconstraint1}), components at lower levels generally merge earlier than their counterparts at higher levels, primarily because the edge costs between components are smaller at lower levels. For a component containing $u$ in level $l$, if $l<l'$, then 
\[
C_{l}(t, u) \supseteq  C_{l'}(t, u).
\]
\item Regarding vertex constraints~(\ref{eq:maxconstraint2}), components at higher levels tend to freeze before their counterparts at lower levels. This occurs because the corresponding dual value at a higher level ceases to increase, due to one of two reasons: either the component is frozen at all higher levels, or it becomes inactive at a higher level.  
\end{enumerate} 
\end{observation}

\paragraph{Pruning Procedure}  The pruning procedure is implemented over $h$ iterations. At the start of the pruning procedure, set $l:= 1$ and let $F_1 = F_1(t_f)$.

\begin{enumerate}
    \item From each tree of $F_l$, remove the maximal pendent-frozen subgraph of $F_l$ and the frozen components from it. The resultant pruned forest is $\tilde{F}_l$. 
    \item If $l:= h$, stop. Otherwise, let $F_{l+1}$ be the union of discarded pendent-frozen subgraphs and components from $F_l$ in the previous step. Set $l = l + 1$ and go to step 1.
\end{enumerate}

Similar to other primal-dual algorithms, frozen components contain vertices connected to vertices in higher level graphs, 
and hence the pruning step discards them for processing in an appropriate (higher) levels. 
Similarly, maximal pendent-frozen subgraphs must be pruned to ensure that no frozen 
component in this layer contributes a degree of one in the standard degree-based inductive 
argument for the $2$-approximation ratio in the primal-dual method.

\paragraph{Feasibility}
We first demonstrate that the resulting pruned forests
constitute a feasible solution to the Tree Combination Problem.
To substantiate this, we show that after any iteration $l$ in the pruning procedure, all vertices $\bigcup_{l' \in [l]} V_{l'}$ from levels $1, 2, ..., l$ and some vertices of $\bar{V}_l$ are connected to exactly one of the resulting forests $\{\tilde{F}_{l'}\}_{l'\in [l]}$, and the remaining vertices discarded in iteration $l$ are either included in an inactive component or found in frozen components at higher levels. 

It is straightforward to verify that this holds true for level $l=1$, as all vertices are either included in the inactive forest $F_1$ or contained within some frozen components. Assume that the lemma holds for $l = l'$. We will now demonstrate that the lemma is also valid for $l = l'+1$. Consider a pruned frozen component $C$ from level $l'$. For any vertex $u \in C$ at any time $t$, the component $C_{l}(t, u) $ is a subset of $ C $ and its
 potential becomes $0$ in level $l$ earlier than component  $C$ in level $l'$. This implies that either
$ C_{l}(t_f, u) \text{ is a frozen component in level $l$} $ and discarded in level $l$ 
or
$ C_{l}(t_f, u) $ is inactive, meaning that $ C_{l}(t, u) \text{ contains a vertex of } V_{l} $. Thus, the components that are pruned away in each pruning step do not affect its connectivity in the higher levels. Hence, the induction step is proved.

\paragraph{Approximation Guarantee}
We first show that the dual value of the LP relaxation can be equivalently written as the sum
of the dual values and penalty costs corresponding to vertices at each level. This will allow us to
bound the cost of the edges in the forest of each level with respect to the dual value.

Consider any level  $l \in\{1, ... , h-1\}$. 
Let $\tilde{V}_l$ denote the set of vertices that are spanned by the forest $\tilde{F}_l$, and let $U_l$ be the union of the vertex sets $\tilde{V}_{l+1}, \tilde{V}_{l+2}, \ldots, \tilde{V}_h$ and $U_0 = V$.   
We further define $X_l$ to be the set of components not spanned by the resulting forests $\tilde{F}_1, \tilde{F}_2, \ldots, \tilde{F}_l$. Notice that the union of the vertex sets of the components in $X_l$ is equal to $U_l$, since $U_l = \bigcup_{l' > l} \tilde{V}_{l'}$. 

We first claim that $X_l$ can be partitioned into sets $X_{l, 1}, X_{l, 2}, ...X_{l, m_l} $ such that $P(Y_l(X_{l, j}), \pi)=0$
 for each set $X_{l, j}$.   To see this, note that for any vertex in $X_{l, j}$ either it was connected to a vertex in the level $l$ root set $V_l$  before the deletion process, or it was not. If it was not connected to the root set, then it was in some frozen component $I$ not containing a root vertex, and $I$ was frozen since its potential $ P(Y_l(I), \pi)=0$. If it was connected to a root vertex before the edge deletion process, then it must have been pruned because it is a pendent-frozen subgraph. According to the notation, its remaining potential $P(Y_l(X_{l,j}), \pi)=0$. Hence, the claim follows.

\begin{lemma} 
\label{lem:dualvaluerelation}
 $$ \sum_{S\subseteq V} Y_{1}(S) = \sum_{S\subseteq V, S\not\subseteq U_1} Y_{1}(S) + \sum_{l=2}^{h-1} \sum_{S\subseteq U_{l-1}, S\not\subseteq U_l} Y_{l}(S)  + \sum_{l=1}^{h-1}  \pi_l(U_l)   $$ 
\end{lemma}
\begin{proof}
Note that
 $$ \sum_{S\subseteq V} Y_{1}(S) =\sum_{S\subseteq V, S\not\subseteq U_1} Y_{1}(S) + \sum_{ S\in U_1} Y_{1}(S).$$ For any  $l = 1, ..., h-1$, 
$$  \sum_{S \subseteq U_{l}}  Y_{l}(S) = \sum_{j=1}^{m_{l}} \sum_{S\subseteq X_{l, j}} Y_{l}(S).$$
Since $X_{l, j}$ is a pruned frozen component from level $l$, therefore its corresponding vertex constraint in~(\ref{eq:maxconstraint2}) is tight. Thus, we have 
 $$ \sum_{S \subseteq U_{l}}  Y_{l}(S) - \sum_{S \subseteq U_{l}\setminus V_{l+1}}  Y_{l+1}(S)   = \pi_{l}(U_{l}). $$
Since any component that includes a vertex from $V_{l+1}$ at level $l+1$ becomes inactive and its dual value is $0$, it follows that 
$$ \sum_{S \subseteq U_{l}}  Y_{l}(S) - \sum_{S \subseteq U_{l}} Y_{l+1}(S)   = \pi_{l}(U_{l}). $$
$$ \Rightarrow   \sum_{S \subseteq U_{l}}  Y_{l}(S) =\sum_{S \subseteq U_{l+1}}  Y_{l+1}(S) +\sum_{S \subseteq U_{l}, S\not\subseteq U_{l+1}}  Y_{l+1}(S)  + \pi_{l}(U_{l}).
 $$

Note that $Y_h(S) = 0$ for any subset $S \subseteq V_h$, since all components at level $h$ are singletons and inactive. 
Applying the above relation recursively, the lemma follows.
\end{proof}

Similar to the standard degree argument used in the primal-dual algorithm proof for the Prize-collecting Steiner Forest problem~\cite{goemans1995general}, for any $l = 1, ..., h-1$, we have   
\begin{align}\sum_{e\in \tilde{E}_l} \cost_l(e)\le 2\sum_{S\subseteq  U_{l-1}, S\not\subseteq U_l} Y_l(S).\label{equ:degreeargument}
\end{align}
 The 
approximation ratio will readily follow from these results.
 
\begin{lemma} 
The primal-dual algorithm constructs a set of forests $\tilde{F}_l =(\tilde{V}_l, \tilde{E}_l)$ such that
$$ \sum_{l\in[h-1]}\sum_{e\in \tilde{E}_l} \cost_l(e) +\sum_{l\in[h-1]} \pi_l(U_l)  \le 2\sum_{S  \subseteq V} Y_1(S)$$
where $U_l = \bigcup_{l' > l} \tilde{V}_{l'}$. 
\end{lemma}
\begin{proof}

 We now rewrite the cost of the primal solution in terms of the dual variables following Equation~(\ref{equ:degreeargument}):
 
\begin{align*}
&\sum_{l\in[h-1]}\sum_{e\in \tilde{E}_l} \cost_l(e) +\sum_{l\in[h-1]} \pi_l(U_l)  &&\le 2\sum_{l\in[h-1]} \sum_{S\subset U_{l-1}, S\not\subseteq U_l} Y_l(S) +\sum_{l\in[h-1]} \pi_l(U_l)  \\
& &&=2 \sum_{S\subseteq V} Y_{1}(S) -2\sum_{l\in[h-1]} \pi_l(U_l) +\sum_{l\in[h-1]} \pi_l(U_l) 
\\
& &&\le 2 \sum_{S\subseteq V} Y_{1}(S)
\end{align*}
Note that the second equality follows from Lemma~\ref{lem:dualvaluerelation}.
\end{proof}

Note that the cost of the constructed solution to the Tree Combination Problem equals
the cost of the IP solution to the IP obtained via the primal-dual algorithm plus
the non-negative fixed cost $\sum_{v\in V} w(v)/p_1$. Therefore, we get a $2$-approximation algorithm for the Tree Combination Problem with multiple different speeds.  Combining this result with Lemma~\ref{lem:ReducingMin-sumCDTtoCombiantion}, we
conclude that there is an $12$-approximation algorithm for the non-preemptive Min-Sum CD
problem.

\begin{theorem}[Min-Sum CD]
\label{thm:Non-preemptiveMin-sumCD}
There exists an $O(1)$-approximation algorithm for the non-preemptive Min-Sum CD problem. 
\end{theorem}

\section{Algorithmic Adjustments for Implementation}\label{sec:adjustments}
Before presenting the experimental results, we discuss two modifications that
we have made in the implementation of our algorithm to improve the performance,
both in terms of running time and solution quality, in practical problem settings.
First, we refine the algorithm for tree construction (Algorithm~\ref{alg:CreatCombiantion}). A straightforward implementation
of step 2 of Algorithm~\ref{alg:CreatCombiantion} would involve constructing
the graph $G'$ of size $O(n^2)$ and computing an MST in $O(n^2)$ time.
In our experiments we only consider instances where all
input points lie in the Euclidean plane. Then the
MST is a subset of the Delaunay triangulation. We therefore compute
a Delaunay triangulation in $O(n\log n)$ time \cite{DBLP:books/lib/BergCKO08} before contracting the
depots into a single depot~$r$ (referred to as \emph{super-depot} in the remainder
of this section). The resulting graph has $O(n)$ edges,
and an MST can be computed in $O(n\log n)$ time. This significantly
reduces the running time.

Furthermore, we observed in the experiments that computing first an MST of the sources
and then attaching the targets produces trees that lead to solutions of relatively large cost.
This is because the vehicles move from the target of one request to the source of the
next request they serve, and therefore target-to-source distances are more relevant
to the application than source-to-source distances. We would therefore like to
grow the spanning tree in a Prim-like fashion by adding to the tree the request whose source is closest to the \emph{target} of a request that is already in the tree, rather than the request whose source is closest to a \emph{source} already in the tree. A straightforward implementation of this idea would increase the running-time
significantly, however, as after a request is added to the tree, the priority of
all sources of requests not yet in the tree may need to be updated. To
avoid this computational expense, we use the Delaunay triangulation to
determine for each target a (small) set of nearby candidate sources, and we update
only the priority of those candidate sources.

As the total length of a spanning tree constructed in this way might not be
within a constant factor of the MST of sources (with targets attached) on
some pathological input instances, we consider a further modification that
guarantees that the spanning tree constructed is within a constant factor of the
MST of sources (with targets attached): When deciding
which request to add to the tree, we add the request with minimum target-to-source
distance $d_{ts}$ only if $d_{ts}$ is at most $k$ times the minimum
source-to-source distance $d_{ss}$ between the tree and a request not in the tree,
where $k$ is a small constant. This ensures that the spanning tree constructed
has total length at most $k$ times the total length of an MST of the sources
(and super-depot) with the targets attached. We add `MST $k$' in brackets
after the name of an algorithm to indicate that this modification has been
applied.

Second, we consider three alternatives for converting the large trees
produced as solution to the Tree Combination problem
by our primal-dual algorithm into tours. In the theoretical analysis,
we had constructed the tours as Euler tours of the large trees.
In the implementation, we perform a DFS of the large tree and
let the vehicle serve the requests in the order in which their
sources are first visited by the DFS. The length of a tour constructed
in this way is within a constant factor of the total length of the large tree.
We refer to the algorithm with this tour construction method as PD-DFS.

As an alternative, we implemented the following cheapest-insertion
heuristic: Build the tour for a large tree incrementally. Initially,
the tour consists only of the depot of the fastest vehicle. Then
process the requests in the large tree in arbitrary order, and
insert each request at the position of the tour where it leads to
the smallest increase in length (either right after the depot or
right after the target of a request already in the tour). We
refer to this method as PD-Greedy. It is efficient as long as
the number of requests included in a large tree is not too large.

As a final alternative, we implemented a two-stage version of the
cheapest-insertion heuristic: First, use the cheapest-insertion
heuristic to build a tour for each of the original trees included
in the large tree. Then, use the cheapest insertion heuristic again
to produce an order of the tours produced for the original trees (discarding
their depots), starting with the tour of the original tree of the fastest depot.
We refer to this method as PD-DGreedy (``double-greedy'').

We note that PD-DFS (MST $k$) is guaranteed to produce a solution that
is a constant-factor approximation of the optimal solution, while the
other algorithm variants do not guarantee this property for all inputs.

In the following, we describe the above-mentioned adjustments in full detail.

\paragraph{Adjustment 1: Minimum-Cost Tree Construction}

The tree construction algorithm (Algorithm~\ref{alg:CreatCombiantion}) is adjusted as follows: We build the tree in a Prim-like fashion. We allow each $(s,t)$ pair also to be connected to the destination of another package in the tree rather than only to the super-depot or a source that is already in the tree.
The packages that are not yet in the tree are maintained in a priority queue,
and the priority of a package $(s,t)$ in the priority queue is
the minimum distance between $s$ to the super depot or any destination already in the tree. If the MST $k$ modification is used, the minimum distance between $s$ and
a source already contained in the tree, multiplied with factor~$k$, is also included in the computation of the priority.
At any time, the package with minimum priority is then removed from the priority queue and added to the tree.
A straightforward implementation of this idea would cause the running time to increase substantially, however, as adding one $(s,t)$ pair to the tree could require the priority of all remaining packages in the priority queue to be updated.
To circumvent this, we determine a set of \emph{candidate sources} for each destination
as follows. Assuming that the input points are points in the Euclidean plane,
we first compute a Delaunay triangulation \cite{DBLP:books/lib/BergCKO08} of the sources. For each destination, we then determine a small set of candidate sources that are `near' the destination, as determined via the Delaunay triangulation.
When an $(s,t)$-pair gets added to the tree, we then only update the distances
of the candidate sources of~$t$, leading to a significantly reduced running time. (If the MST~$k$ modification is used, we also update the Delaunay neighbors of
$s$.)

\begin{enumerate}[label=Step \arabic*:]
  \item \textbf{Super-depot contraction.}  
  All depots are contracted into a virtual node $r_{\text{super}}$, the super depot. For each source node, its distance to $r_{\text{super}}$ is defined as the distance to the nearest depot.  

  \item \textbf{Source-only Delaunay triangulation \footnote{We also experimented with constructing a target-only Delaunay triangulation (and using it to determine for each source the targets to which it is near), allowing each source to have one or more opportunities for its connection cost to be updated. However, in tests with uniformly distributed points, this approach produced higher final costs compared to the source-only Delaunay method.}.}  
  A Delaunay triangulation is constructed using only the source nodes, producing a sparse set of edges among them.  

  \item \textbf{Determine candidate sources (see Figure~\ref{fig:case_analysis_for_target_insertion_fig}).}  
  For each target node $t$:  
  \begin{enumerate}[label=Step 3.\arabic*:]
    \item If $t$ lies inside the circumcircle of any Delaunay triangle, the corresponding triangle is considered relevant to $t$ and its three vertices are added to the set of candidate sources for $t$.  
    \item If $t$ does not lie inside any circumcircle, two tangent lines are drawn from $t$ to the convex hull of the sources. The chain of hull vertices between the tangent points is taken as the candidate source set.  
    \item If the source $s$ corresponding to $t$ (i.e., the source $s$ of the package that is to be delivered to $t$) is contained in the candidate source set, then all neighbors of $s$ in the Delaunay graph are also added to the set.  
  \end{enumerate}

  \item \textbf{Prim-based tree construction.}  
  A Prim-like algorithm is executed on the sparse graph consisting of $r_{\text{super}}$, the sources $S$, and the targets $T$:  
  \begin{enumerate}[label=Step 4.\arabic*:]
    \item The connection cost of each source is initialized as its distance to $r_{\text{super}}$ (Step~1).  
    \item The source with the minimum connection cost is selected and connected to the tree.  
    \item The paired target $t$ of the selected source is immediately connected.  
    \item (a) The connection costs of the candidate sources associated with this target (Step~3) are updated. More precisely, for every candidate source $s$ that is not yet in the tree and has current priority $p$, its priority is updated to $\min\{p,d(t,s)\}$.
    
    (b) If the MST $k$ modification is used: The connection costs of all Delaunay neighbors of the source selected in Step 4.2 are updated. Specifically, for every source $s'$ that is not yet in the tree and is a Delaunay neighbor of the selected source $s^*$, if its current priority is $p'$, then its priority is updated to $\min\{p', k \cdot d(s^*,s')\}$, where $k$ is a fixed constant.
    \item Steps~4.2--4.4 are repeated until all $(s,t)$ pairs are connected into the tree.  
  \end{enumerate}

  \item \textbf{Super-depot expansion.}  
  The virtual node $r_{\text{super}}$ is replaced by the original depots.  

  \item  Return the collection of trees obtained.
\end{enumerate}

Steps~2 and~3 play a critical role in reducing the computational burden: once a new $(s,t)$ pair is connected, only the sources in the vicinity of the corresponding target, and potentially the sources that are Delaunay neighbors of~$s$, require cost updates,  rather than all sources.

Step 4.4 can be instantiated in two ways. The first variant applies only rule (a) and updates connection costs solely via target to source distances.
This variant follows the problem structure by prioritizing target-to-source connections, but it does not guarantee a constant approximation ratio. In particular, a subtle difficulty arises when certain sources are distant from all targets: In such cases, the connection costs of these sources may not be updated during the process and thus remain at their initial values, defined as the distance to the virtual super-depot $r_{\text{super}}$.
The second variant, called the MST $k$ modification, applies both rules (a) and (b). In this case, we maintain two priorities for each remaining source: (i) its best source–target priority $P_{ST}$ (the minimum distance from targets currently in the tree, initialized as the distance from the super-depot), and (ii) its best source–source priority $P_{SS}$ (to the sources that are Delaunay neighbors). If $P_{ST} \le k \cdot P_{SS}$, we add the request corresponding to $P_{ST}$; otherwise, we add the request indicated by $P_{SS}$. Here, $k$~is a fixed constant. This additional consideration of source–source distances is essential for guaranteeing a constant-factor approximation ratio.

\begin{figure}[H]
\centering
\includegraphics[width=0.85\textwidth]{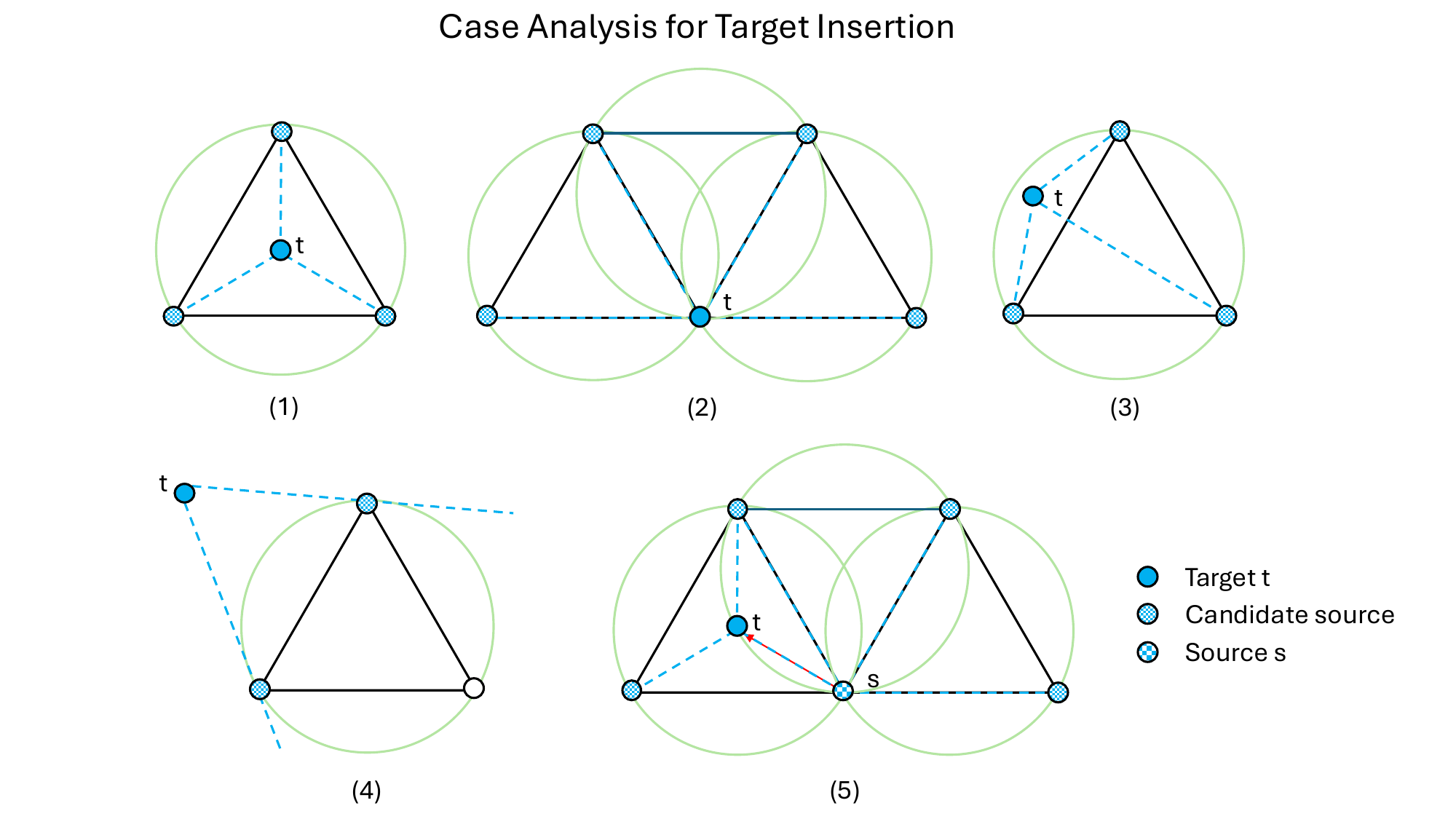}
\caption{The cases for determining candidate sources: (1)–(3) Step 3.1: The target $t$ lies inside the circumcircle of any Delaunay triangle. (4) Step 3.2: The target $t$ lies outside all circumcircles. (5) Step 3.3 : The corresponding source $s$ of $t$ is contained in the candidate source set.}  
\label{fig:case_analysis_for_target_insertion_fig}
\end{figure} 

\paragraph{Adjustment 2: Heuristic Final Path Construction} 
The result of the primal-dual procedure is a set of large trees, each of which is a combination of one or several original trees. We refer to the original trees
also as \emph{small trees} in the following. Each of the small trees is associated with a depot, and for each large tree we assign the fastest depot (vehicle) of its small trees to serve all $(s,t)$ pairs in the large tree.

We consider three different methods for constructing the final tour for each large tree.
Based on the approach used in the proof of Lemma~\ref{lem:ReducingMin-sumCDTtoCombiantion}, our first method is to perform a depth-first search (DFS) traversal on each large tree to construct a feasible final path. This procedure is computationally efficient, and the total cost of the resulting path is within a constant factor of the tree cost, providing the desired performance guarantee. We refer to this implementation as the \emph{Primal--Dual Algorithm with DFS} (PD-DFS).

The second method is to use the greedy cheapest-insertion heuristic: For each large
tree, we construct a route for its fastest vehicle to serve all requests in the
large tree. As this is done separately for each large tree, the need to compare against different depot paths and possible insertion positions is eliminated. This substantially reduces the computational cost and explains why our algorithm performs significantly faster than the heuristic baseline (see Section~\ref{sec:exp}). We refer to this approach as the \emph{Primal--Dual Algorithm with Greedy Routing} (PD-Greedy).

Additional efficiency is achieved by exploiting the fact that the requests in
a large tree are partitioned among the small trees that make up the large tree.
Specifically, we first apply the cheapest-insertion heuristic separately to the set of packages in each small tree (starting with the depot that was originally associated with that small tree). We then apply the heuristic a second time to determine the order in which the resulting paths (starting with the path of the small tree that contains the depot with the fastest drone, and omitting the depots of the paths of all
other small trees) are combined into the final path. If both the number of packages in each small tree and the number of small trees within each large tree are small, the greedy cheapest-insertion heuristic is applied only to instances of small size, resulting in a reduced running time. We refer to this approach as the \emph{Primal--Dual Algorithm with Double Greedy Routing} (PD-DGreedy).

By adjusting the tree construction procedure and considering three different methods for final path construction, we thus have three variants of our PD-based algorithm:
\begin{itemize}
\item \emph{Primal--Dual algorithm with Depth-First Search Routing} (PD-DFS);
\item \emph{Primal--Dual algorithm with Greedy Routing} (PD-Greedy); 
\item \emph{Primal--Dual algorithm with Double Greedy Routing} (PD-DGreedy).
\end{itemize}
Furthermore, for each of these algorithms we can optionally apply the MST $k$ modification to the tree construction step. In our experiments we found that
the choice $k=7$ worked best, and we refer to these three variants
as PD-DFS (MST 7), PD-Greedy (MST 7), and PD-DGreedy (MST 7).

\paragraph{Running Time Analysis}
For ease of presentation of the following analysis, we make the natural assumption that $k\le n$.
In the tree-construction phase, we first construct a Delaunay triangulation of the source nodes. This step runs in \(O(n \log n)\) time~\cite{DBLP:books/lib/BergCKO08}.
We then determine for each target node a set of candidate sources; this also takes
$O(n \log n)$ time provided that the number of candidate sources for each target
node is constant on average, which is typically the case. The Prim-based construction
of a spanning tree then also runs in $O(n \log n)$ time.

In our implementation of the primal-dual procedure, as analyzed in Lemma~\ref{lem:primal-dual_runningtime}, the main loop of the algorithm terminates after 
$
\sum_{l \in [h-1]} |V_l| \cdot (2l-1)
$
iterations, where \(V_l\) denotes the set of constructed trees whose depots are at level \(l\). Note that $|V_l|\le k$. The number of iterations of the primal-dual procedure thus depends
only on $k$ and~$h$, independent of the number of requests.

When the number of request pairs is small, the primal-dual phase dominates the running time. When the number of request pairs is large, the computation of the final routes becomes
more computationally expensive, especially for PD-Greedy and PD-DGreedy.
Specifically, the routing part of the PD-DFS algorithm requires $O(\sum_i N_i)=O(n)$ time, where \(N_i\) denotes the number of source-target pairs contained in the \(i\)-th large tree produced by the primal--dual algorithm.
The routing part of the PD-Greedy algorithm requires 
$
O\!\left(\sum_i N_i^2\right)
$
time, where \(N_i\) denotes the number of source-target pairs contained in the \(i\)-th large tree produced by the primal-dual algorithm. The routing part of the PD-DGreedy algorithm requires
$
O\!\left(\sum_i \left(S_i^2 + \sum_j M_{i,j}^2\right)\right)
$
time, where $S_i$ denotes the number of small trees contained in the $i$-th large
tree and \(M_{i,j}\) denotes the number of source-target pairs in the \(j\)-th small tree contained in the \(i\)-th large tree.

The algorithmic adjustments discussed in this section partly improve the running time
and partly improve the solution quality at the cost of a small increase in running time.
In the next section, we evaluate the performance and running time of our PD-based algorithms relative to a baseline algorithm.

\section{Computational Experiments}
\label{sec:exp}
In Section~\ref{sec:InstanceAndBaseline}, we describe the synthetic and
real-word instances used for our experiments as well as the baseline algorithm
to which we compare our primal-dual based algorithms. Then, in Section~\ref{sec:expeval},
we present the results of our experiments. Our implementations, which allow all experimental results to be fully reproduced, are publicly available at \url{https://github.com/WenZhang-Vivien/Tree_construction_Primal_Dual_Routing}

\subsection{Instances and Baseline Algorithm}
\label{sec:InstanceAndBaseline}

\paragraph{Synthetic Datasets}
We consider three synthetic datasets in our experimental evaluation.
The first dataset is constructed based on the worst-case instance described below for the baseline algorithm with which we compare the performance of our algorithms. It is used to demonstrate the guaranteed performance of our algorithms relative to the baseline. 

The second dataset, denoted by SY-U, consists of instances in which all locations---including request pickup locations, drop-off locations, and depot locations--—are independently generated from a uniform distribution over a $[0,100]\times[0,100]$ grid.

The third dataset, denoted by SY-GMM, is generated using a Gaussian mixture model with $c$ clusters. Each cluster follows a bivariate Gaussian distribution whose center is sampled uniformly from the $[0,100]\times[0,100]$ grid and whose covariance matrix is given by $\sigma^2 I$. To study the effects of clustering structure and spatial dispersion, we evaluate multiple combinations of the parameters $c$ and $\sigma$, as summarized in Table~\ref{tab:dataset_parameters}.

 \begin{table}[ht]
\centering
\caption{Parameter settings for the synthetic datasets. Bold values indicate fixed parameters used in the sensitivity analyses, while ``--'' denotes parameters that are not applicable.}
\label{tab:dataset_parameters}
\begin{tabular}{lcc}
\toprule
\textbf{Parameter} & \textbf{Worst-case} & \textbf{SY-U} / \textbf{SY-GMM} \\
\midrule
$n$ (number of requests) 
& $5,10,15,20 \,(\times 10^3)$ 
& $5,\mathbf{10},15,20 \,(\times 10^3)$   \\

$k$ (number of depots) 
& $2$ 
& $1,3,5,7,9,10,15,20,\mathbf{30} ,40 \;(\times 3)$  \\

$h$ (number of speeds/levels) 
& $2$ 
& $1,2,\mathbf{3},4,5$  \\

$c$ (number of clusters) 
& -- 
&  --  / $1,2,3,4,5,6 \;(\times 5)$ \\

$\sigma$ (cluster covariance) 
& -- 
&  --  /$1,2,3,4,5,6 \;(\times 5)$ \\
\bottomrule
\end{tabular}
\end{table}

\paragraph{Realworld Datasets}

We also conduct experiments using a real-world dataset collected from Meituan~\cite{meituan_informs_tsl_challenge}, one of the largest on-demand delivery platforms. The dataset contains orders, couriers (we treat the initial locations of the couriers as depots), and delivery records from an anonymized midsize city in China. To capture both regular and peak demand patterns, we select one representative weekday (October 18, 2022) and one representative weekend day (October 22, 2022) for evaluation.

To study scalability under varying system sizes, we consider multiple depot-scale settings with 3, 9, 15, 21, 27, 30, 45, 60, 90, and 120 depots. For each setting, depots are randomly sampled from the full pool of available depots. To ensure controlled comparisons across different scales, depot sets are constructed incrementally: the depot set in a larger-scale setting strictly contains the depot set from the previous smaller-scale setting. For example, the 15-depot configuration includes the same three depots used in the 3-depot setting, together with 12 additional depots randomly sampled from the remaining depots. 

For each depot configuration and each selected day, all orders recorded
on that day in the dataset are used as requests.
In these datasets, each location is represented by its latitude and longitude coordinates. The distance between two locations is defined as the Euclidean distance between their coordinates. A summary of key statistics for all experimental settings, including the number of depots, number of orders,
and other relevant characteristics, is reported in Table~\ref{tab:realdataset_parameters}.

\begin{table}[ht]
\centering
\caption{Parameter settings for the real-world datasets.}
\label{tab:realdataset_parameters}
\begin{tabular}{lcc}
\toprule
\textbf{Parameter} & \textbf{Weekday} & \textbf{Weekend} \\
\midrule
$n$ (number of requests) 
& 69{,}103 
& 77{,}638 \\

$k$ (number of depots) 
& \multicolumn{2}{c}{$1, 3, 5, 7, 9, 10, 15, 20, 30, 40 \;(\times 3)$} \\

$h$ (number of speed levels) 
& \multicolumn{2}{c}{$3$} \\
\bottomrule
\end{tabular}
\end{table}

\paragraph{Baseline Algorithm}
We are not aware of any previously proposed methods for the
heterogeneous collaborative delivery problem.
Heuristic algorithms, such as greedy methods, tend to be computationally efficient and align well with our goal of solving the problem quickly.
Therefore, we consider the following cheapest-insertion heuristic algorithm as the baseline in our experiments.
The algorithm begins with an arbitrary order of requests and processes them sequentially. For each request, it evaluates all vehicles and all insertion positions within their current routes (a request can be inserted after the depot or after the destination of any package that is already contained in the route), computing the additional travel time incurred. The request is then assigned to the vehicle and insertion position that yields the minimum increase in cost. This process is repeated until all requests are placed, producing a feasible set of routes that serves as a baseline for comparison in our experiments. The running time of the baseline algorithm is \(O(n(k+n))\).

We emphasize that this greedy cheapest-insertion heuristic is not a constant-factor approximation, even in the special case where each request is of the form $s_i = t_i$. In fact, we next construct an instance where the algorithm incurs a cost $\Theta(\alpha)$ times the optimal solution, where $\alpha = v_{\max}/v_{\min}$ is the ratio of the maximum to minimum vehicle speed.

\paragraph{Worst-case Instance of the Baseline Algorithm} 
Consider two vehicles with speeds $v_1$ and $v_2 = \alpha v_1$ with $\alpha>1$, located at distinct depots $o_1$ and $o_2$ respectively. 
There are $n$ requests $\{(s_i,t_i)\}_{i=1}^n$ with $s_i = t_i$.
The depots and request endpoints are placed in a line metric as follows:
\begin{itemize}
    \item $o_1$ is placed at $-nv_1$,
    \item $o_2$ is placed at $-(n+\varepsilon)v_2$ for some small $\varepsilon>0$,
    \item $s_i=t_i$ is placed at $(i-1)nv_1$ for $i=1,2,\ldots,n$.
\end{itemize}
By construction, the greedy cheapest-insertion heuristic assigns all requests to the slower vehicle (vehicle~1). 
Hence, the total travel time is
\[
\frac{n \cdot n \cdot v_1}{v_1} = n^2.
\]
On the other hand, if all requests are served by the faster vehicle (vehicle~2), then the total travel time is
\[
\frac{(n+\varepsilon)\cdot v_2 +  (n-1)\cdot n \cdot v_1}{v_2} 
\le n + \varepsilon + n^2 \cdot \frac{v_1}{v_2}.
\]
When $n$ is large, the optimal solution cost approaches $O(n^2 \cdot \frac{v_1}{v_2})$, whereas the greedy algorithm incurs cost $n^2$. 
Thus, the approximation ratio of the greedy heuristic is $\Theta(v_2/v_1) = \Theta(\alpha)$, showing that the algorithm is not a constant-factor approximation.

\subsection{Experimental Evaluation}\label{sec:expeval}
We first present the results for synthetic inputs and then those
for real-world data.
All experiments were conducted on a laptop with
Apple M2 Pro processor and 16 GB main memory running macOS~15.2.

\subsubsection{Computational Results in Synthetic Data Sets}
In this section, we present the results on the worst-case, SY-U, and SY-GMM datasets.
For the total travel time objective, the baseline algorithm performs well in synthetic data sets, but is slowest, and in some instances can perform significantly worse due to the lack of a performance guarantee.    
All three PD-based algorithms run faster than the baseline algorithm. Among them, PD-Greedy consistently achieves performance close to that of the baseline algorithm. PD-DFS is the fastest algorithm, but its solution quality in the experiments is worse than that of the other two. PD-DGreedy offers a good trade-off: it is faster than PD-Greedy while achieving better performance than PD-DFS. 

\paragraph{Computational Results in Worst-Case Data Sets} 
We construct a worst-case dataset on a line metric with two speed levels, following the example described in the previous section. Specifically, we set $\varepsilon$, $v_{\max}$, and $v_{\min}$ as fixed parameters and increase the number of source--target pairs $n$.
We let $\alpha=v_{\max}/v_{\min}=1,000$.
Figure~\ref{fig:routeCost_VS_runningTime_ratio_plot_diff_stPairs} shows the solution
cost and running-time versus the number $n$ of requests.
As $n$ increases, the total cost of the baseline algorithm grows with $n^2$ whereas the costs of all PD-based algorithms are about
$1,000$ times smaller, demonstrating the superior performance of our proposed methods.
The MST 7 variants of the PD algorithms produce the same solutions as the variants not using the MST 7 modification and are not shown.

In terms of running time, all algorithms except PD-DFS exhibit significant increases in running time as the number of requests $n$ increases. The baseline algorithm consistently requires more running time than the PD-based algorithms across all instances.
Compared to the other algorithms, the running time of PD-DFS grows much more slowly as the number of requests increases.

\begin{figure}[H]
\centering
\includegraphics[width=0.85\textwidth]{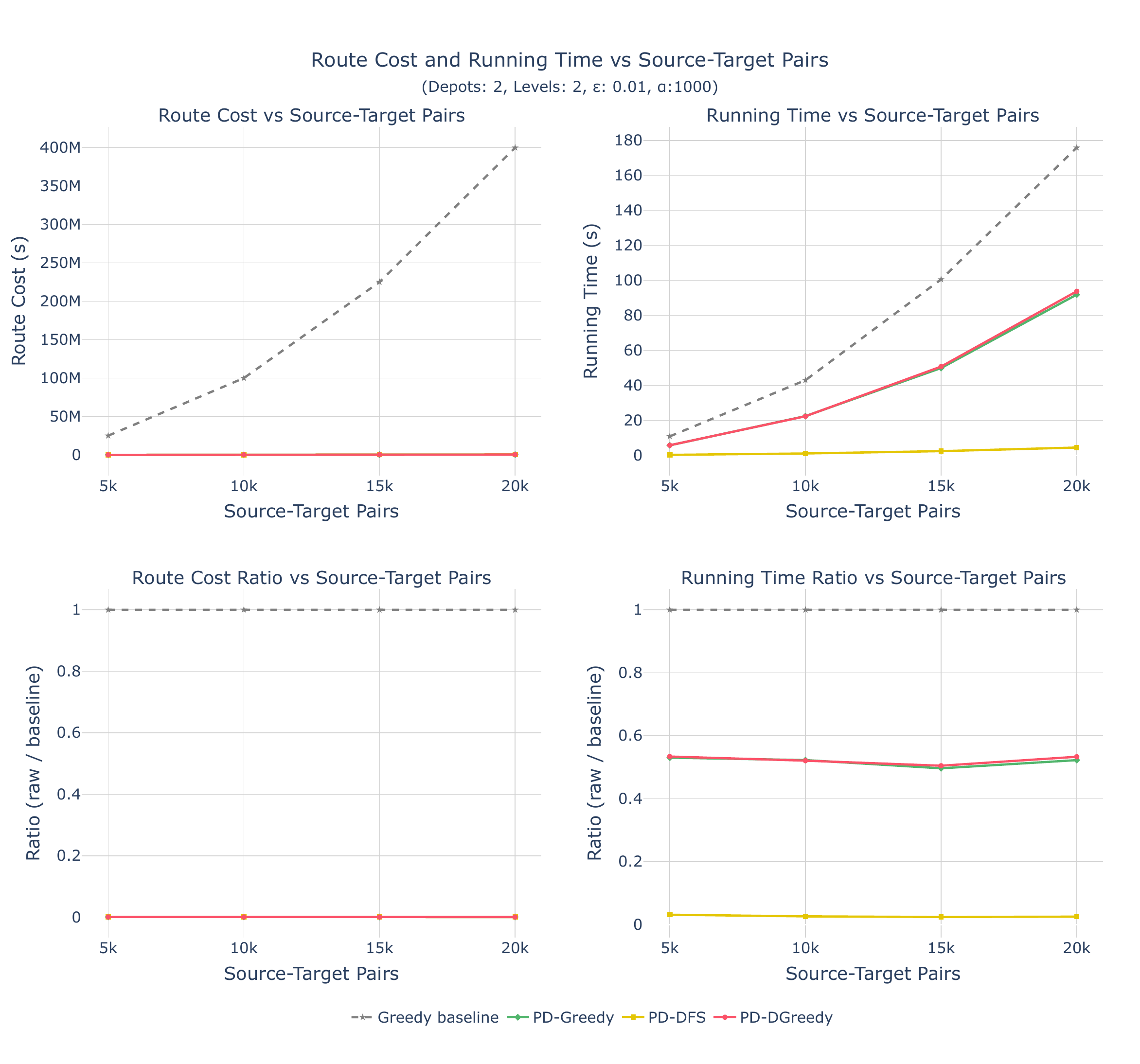}
\caption{Results on the worst-case instance for the baseline algorithm with varying $n$, where $\alpha = v_{\max}/v_{\min}$ and $\epsilon$ is a small constant parameter (see Section~\ref{sec:InstanceAndBaseline}).
}  
\label{fig:routeCost_VS_runningTime_ratio_plot_diff_stPairs}
\end{figure}

\paragraph{Computational Results in SY-U and SY-GMM Data Set}   
Figures~\ref{fig:routeCost_VS_runningTime_ratio_plot_diff_stPair}, 
\ref{fig:routeCost_VS_runningTime_ratio_plot_diff_depots_DFS}, and 
\ref{fig:routeCost_VS_runningTime_ratio_plot_diff_levels_speed_decay_5_DFS}
illustrate the impact of the number of source--target pairs \(n\), the number of depots \(k\), and the number of speeds/levels \(h\) on algorithm performance for data generated from a uniform distribution.
The corresponding figures for
data generated from a GMM distribution can be found in Appendix~\ref{app:SY_GMM_results}. We observe that the performance trends on the SY-GMM datasets are similar to those on the SY-U datasets; therefore, we only provide a detailed discussion of the effects of \(n\), \(k\), and \(h\) for the uniform case. Finally, we examine how the performance and running time vary with respect to the GMM parameters, namely the number of clusters \(c\) and the cluster covariance \(\sigma\).

For a fixed number of depots ($k = 90$) and a fixed number of speed levels ($h = 3$), as the number of source-target pairs $n$ increases from $5\times 10^3$ to $20\times 10^3$ as shown in Figure~\ref{fig:routeCost_VS_runningTime_ratio_plot_diff_stPair}, the ratio between the total travel time objective of our PD-Greedy (and PD-DGreedy) algorithm and that of the baseline algorithm decreases and approaches~1. Meanwhile, the ratio of the running times remains below~1 and continues to decrease, indicating a growing computational advantage of our approach at larger scales.
The MST 7 variants of the PD algorithms tend to produce slightly better solutions at the
expense of a slightly increased running time.

Figure~\ref{fig:TC_Primal_Dual_running_time_with_route_DFS} illustrates the breakdown of the running time of the individual components of each PD-based algorithm as the number of request pairs increases. When \(n\) is small (e.g., \(n = 5{,}000\)), the primal--dual component dominates the running time. As \(n\) increases, the routing phases account for a larger portion of the total running time: first the PD-Greedy routing component becomes dominant (around \(n = 10{,}000\)), followed by the PD-DGreedy routing component (around \(n = 15{,}000\)), and for sufficiently large \(n\) (e.g., \(n = 20{,}000\)), the PD-DFS routing phase also dominates the primal--dual step.

\begin{figure}[H]
\centering
\includegraphics[width=0.85\textwidth]{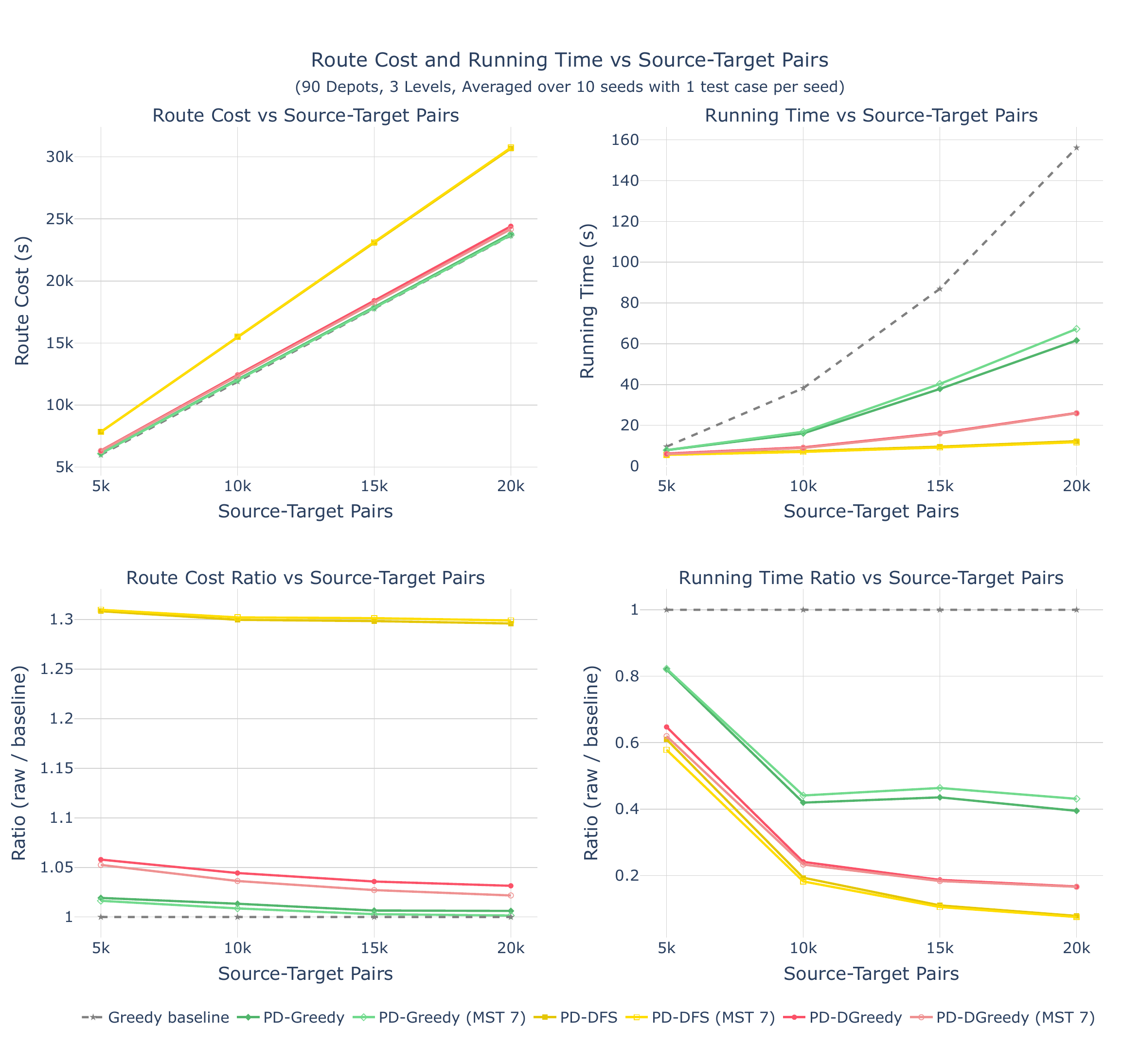}
\caption{Results on the Uniform synthetic dataset (SY-U) with varying $n$.}  
\label{fig:routeCost_VS_runningTime_ratio_plot_diff_stPair}
\end{figure} 

\begin{figure}[H]
\centering
\includegraphics[width=0.85\textwidth]{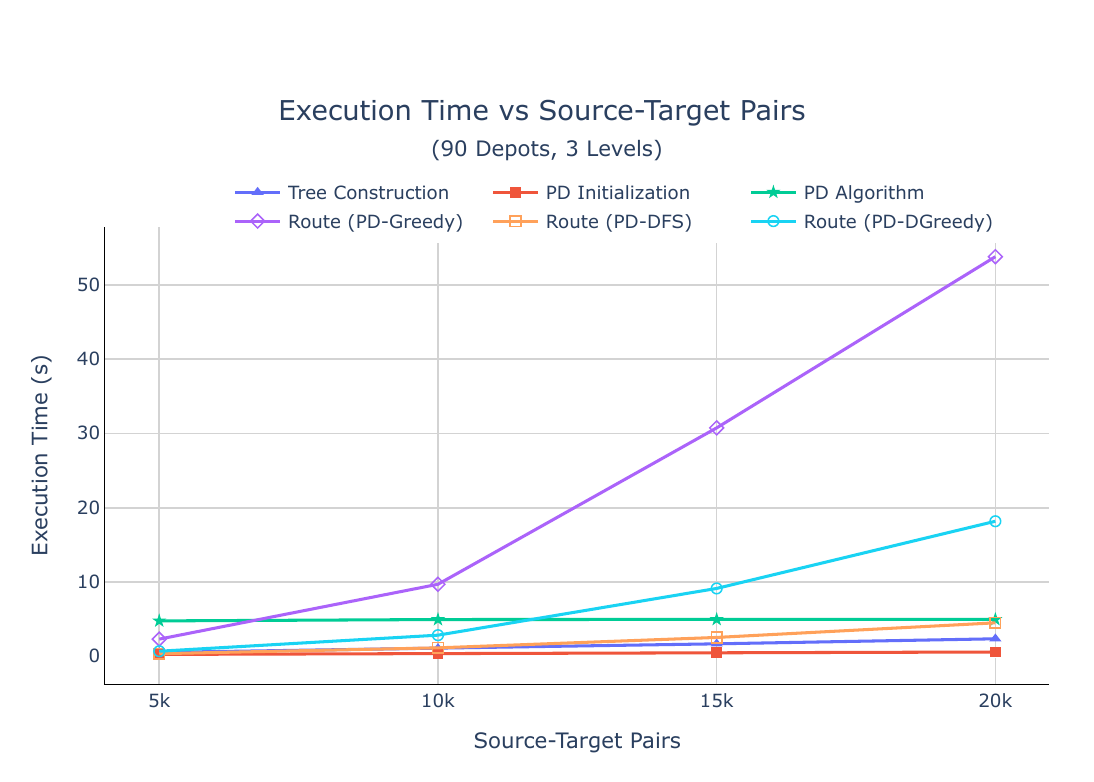}
\caption{Breakdown of running time for the three PD-based algorithms  on the Uniform synthetic dataset (SY-U) with varying $n$.}  
\label{fig:TC_Primal_Dual_running_time_with_route_DFS}
\end{figure}

For a fixed number of source-target pairs ($n = 10\times 10^3$) and a fixed number of speed levels ($h = 3$), when the number of depots $k$ is small (less than~30), as shown in Figure~\ref{fig:routeCost_VS_runningTime_ratio_plot_diff_depots_DFS}, our PD-Greedy (and PD-DGreedy) algorithm outperforms the baseline in terms of solution quality; with an increasing number of depots, the baseline algorithm attains slightly better performance. 
In terms of running time, the PD-based algorithms are always faster than the baseline. Interestingly, the running time of all PD-based algorithms decreases as the number of depots increases when the number of depots is relatively small (less than~60), and then increases thereafter. This behavior can be explained by the changing dominance of different algorithmic components. When the number of depots is small,
 the routing component dominates the overall running time, and its cost decreases as more depots are added. As the number of depots becomes large, the primal--dual (PD) component becomes dominant, leading to increased running time with further increases in the number of depots. This trend is illustrated in Figure~\ref{fig:TC_Primal_Dual_running_time_diff_depots_with_route}, where we separately report the running times of the individual components of each PD-based algorithm as the number of depots varies.  
These observations suggest that, for a fixed number of source-target pairs, there may exist an appropriate number of depots that simultaneously yields best solution quality and fast running time.

\begin{figure}[H]
\centering
\includegraphics[width=0.85\textwidth]{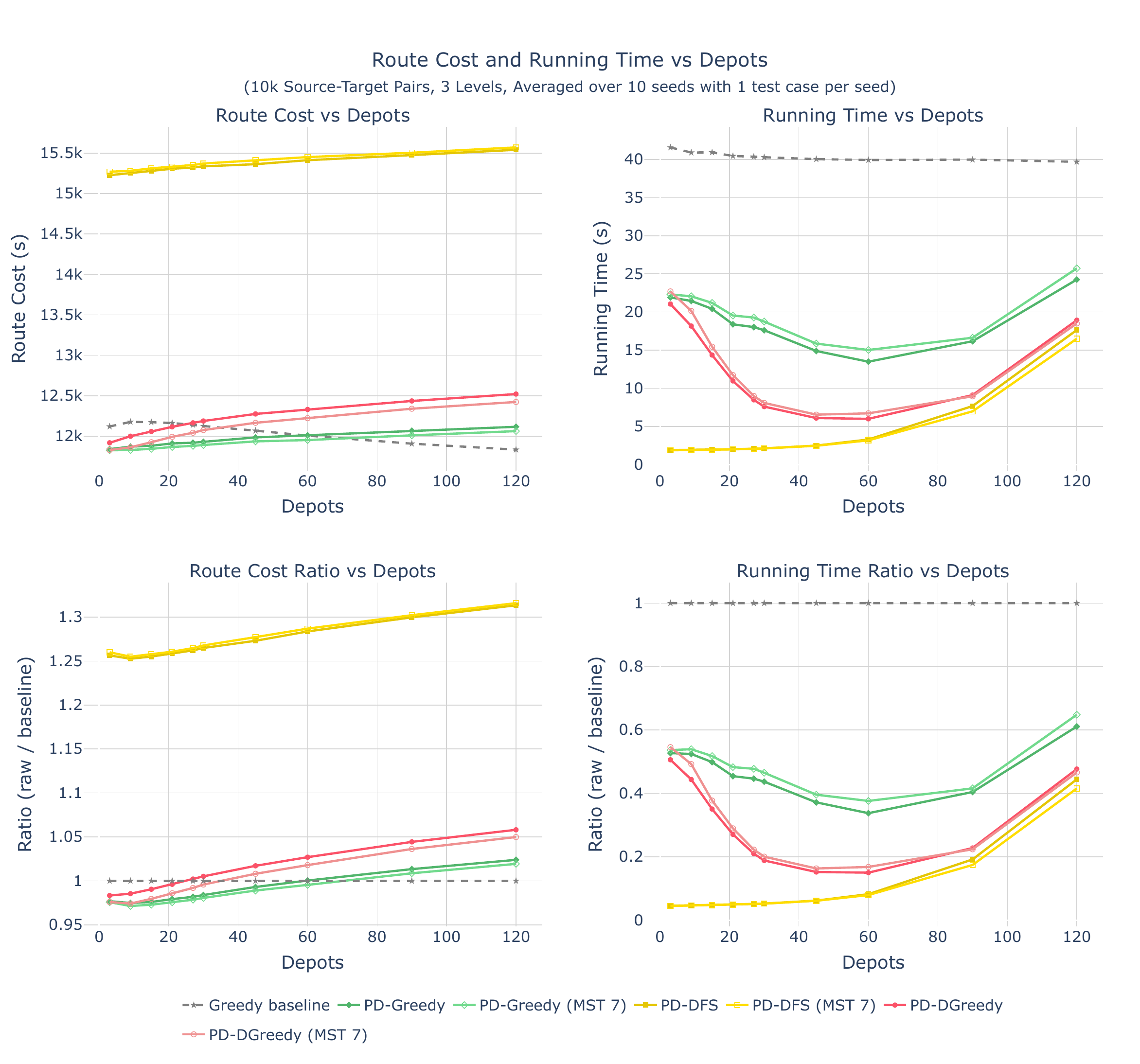}
\caption{Results on the Uniform synthetic dataset (SY-U) with varying $k$.}  
\label{fig:routeCost_VS_runningTime_ratio_plot_diff_depots_DFS}
\end{figure}

\begin{figure}[H]
\centering
\includegraphics[width=0.85\textwidth]{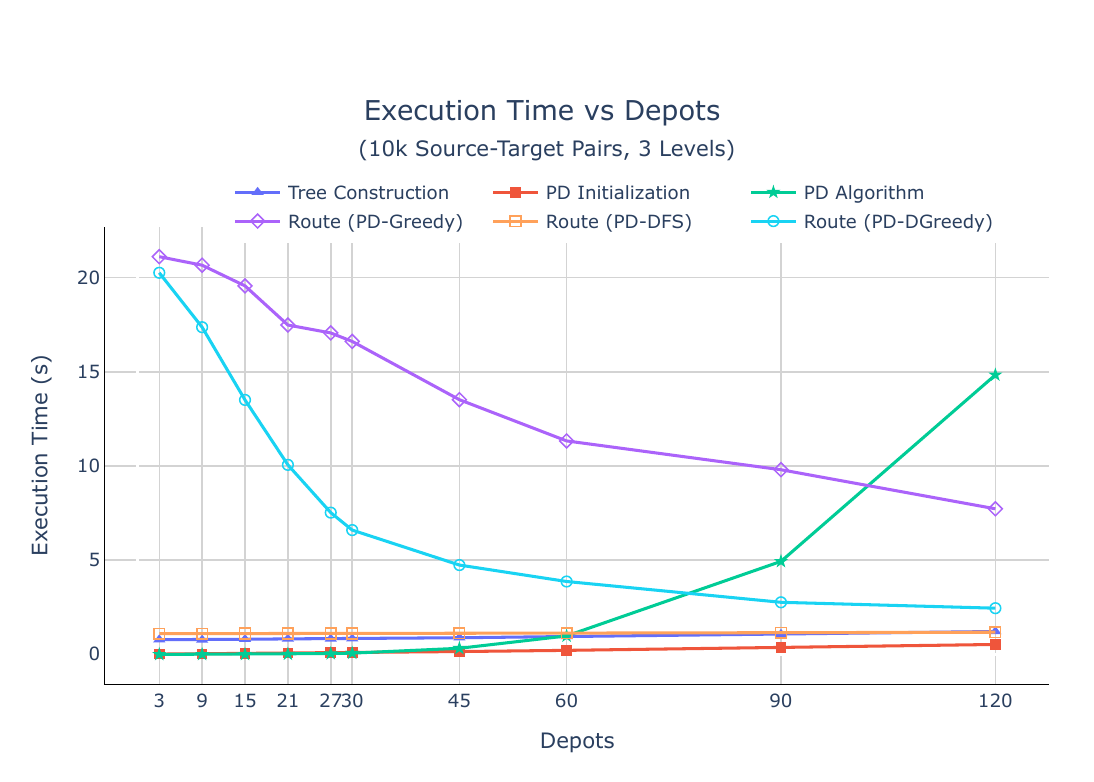}
\caption{Breakdown of running time for the three PD-based algorithms  on the Uniform synthetic dataset (SY-U) with varying $k$.}  
\label{fig:TC_Primal_Dual_running_time_diff_depots_with_route}
\end{figure}

For a fixed number of source-target pairs ($n = 10\times 10^3$) and a fixed number of depots ($k = 90$), as the number of speed levels $h$ increases, as shown in Figure~\ref{fig:routeCost_VS_runningTime_ratio_plot_diff_levels_speed_decay_5_DFS}, 
the performance of our PD-Greedy algorithm is slightly worse than that of the baseline algorithm when the number of speed levels is small (less than~3), but becomes superior as the number of speed levels increases. In terms of running time, the baseline algorithm spends most of its time in the same computational components, and its running time therefore remains largely unchanged as the number of speed levels increases.  
In contrast, the running time of our PD-based algorithms  increases with the number of speed levels, but they consistently remain faster than the baseline algorithm. This increase is mainly due to the Primal Dual component, whose running time grows with the number of speed levels. This trend is illustrated in Figure~\ref{fig:TC_Primal_Dual_running_time_diff_levels_with_route_speedDiff_5}, where we separately report the running times of the individual components of each PD-based algorithm as the number of speed levels varies.  

\begin{figure}[H]
\centering
\includegraphics[width=0.85\textwidth]{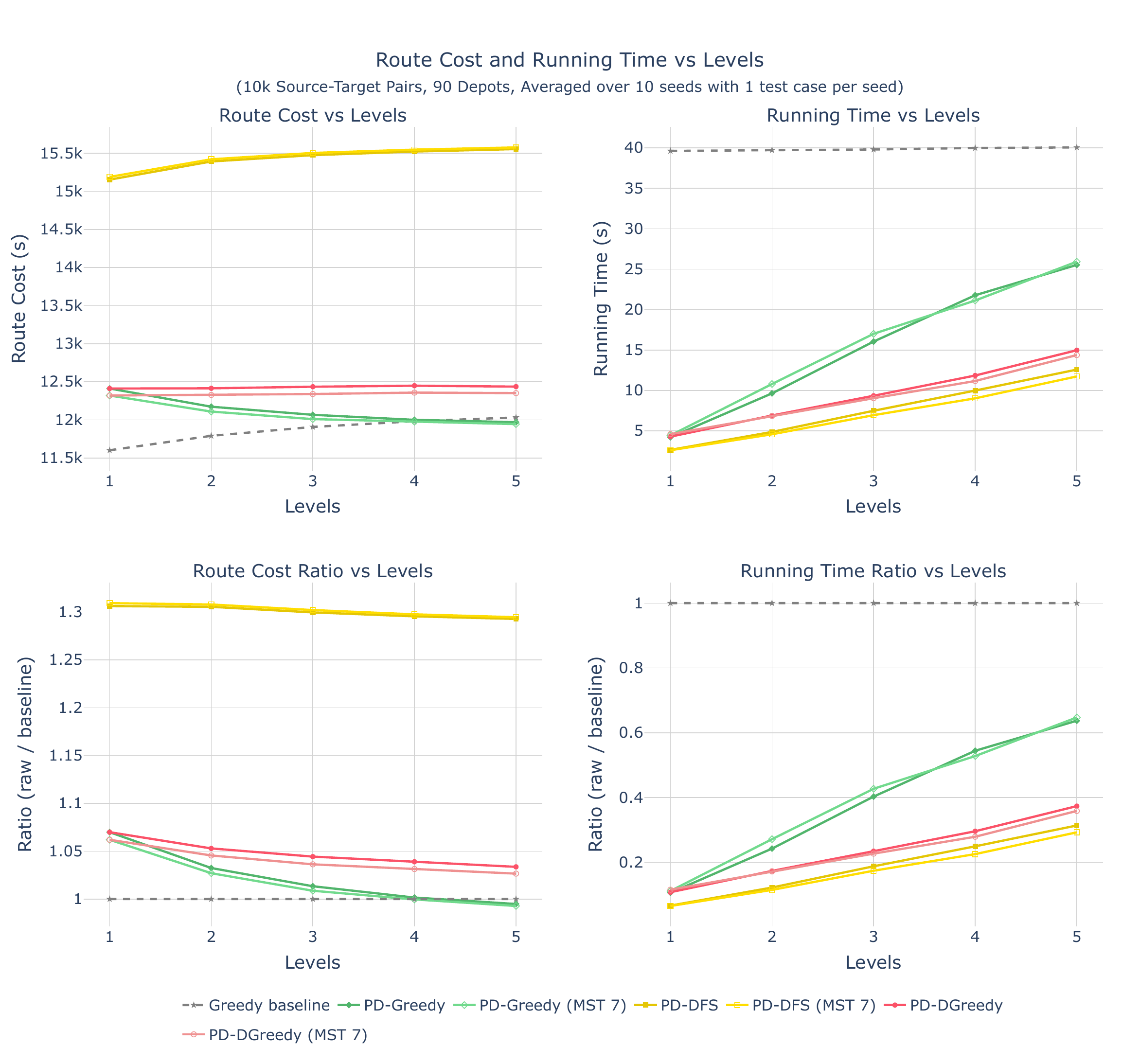}
\caption{Results on the Uniform synthetic dataset (SY-U) with varying $h$.}  
\label{fig:routeCost_VS_runningTime_ratio_plot_diff_levels_speed_decay_5_DFS}
\end{figure}

\begin{figure}[H]
\centering
\includegraphics[width=0.85\textwidth]{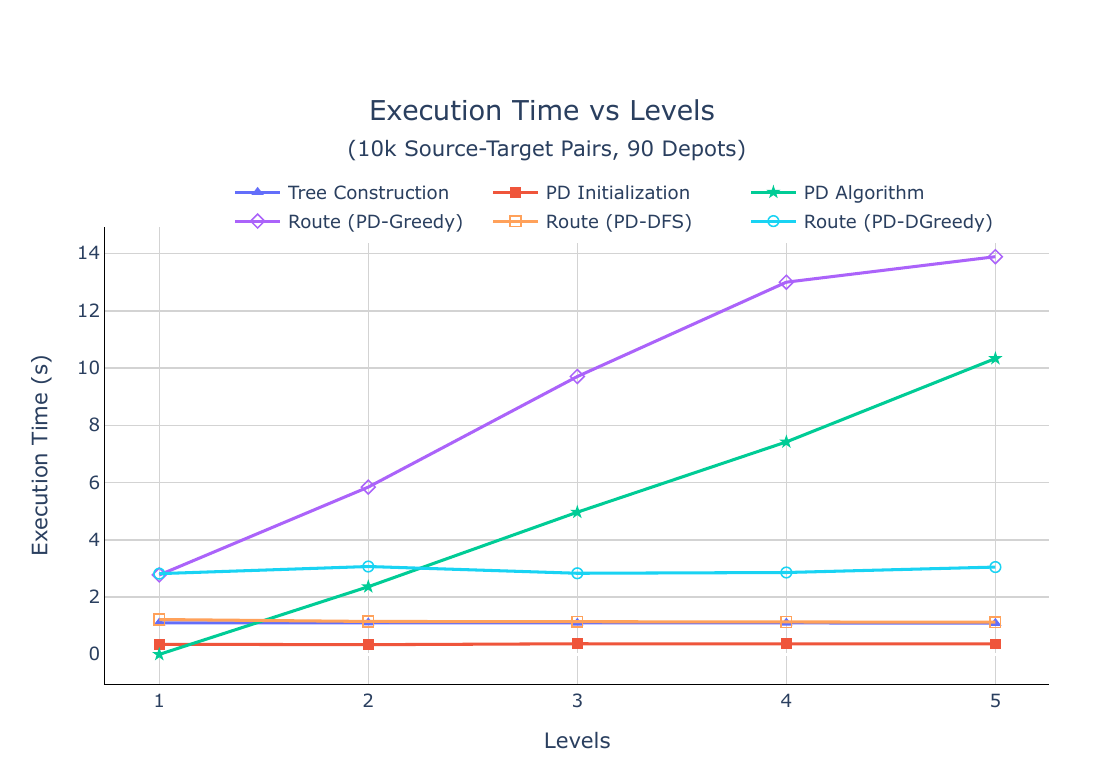}
\caption{Breakdown of running time for the three PD-based algorithms  on the Uniform synthetic dataset (SY-U) with varying  $h$.}  
\label{fig:TC_Primal_Dual_running_time_diff_levels_with_route_speedDiff_5}
\end{figure}

Figure~\ref{fig:route_cost_Radar_Plot_two_panels} illustrates the effect of the GMM distribution parameters—--the number of clusters $c$ and the covariance $\sigma$--—on algorithm performance. All three PD-based algorithms exhibit similar trends, with PD-DFS (yellow curve) clearly illustrating the performance changes with respect to $c$ and $\sigma$ in the figure.

As either $c$ (which partitions source-target pairs into more spatially separated clusters) or $\sigma$ (which increases spatial variance) grows, the data become more spatially dispersed. In this setting, the greedy cheapest-insertion heuristic performs poorly because local cost minimization leads to early routing decisions that create large and irreversible inter-cluster detours. In contrast, the PD-based algorithms leverage early assignment and hierarchical structure to better respect the spatial organization of the data, resulting in increasingly lower route cost ratios and a growing performance advantage over the greedy baseline.

 \begin{figure}[H]
\centering
\includegraphics[width=0.85\textwidth]{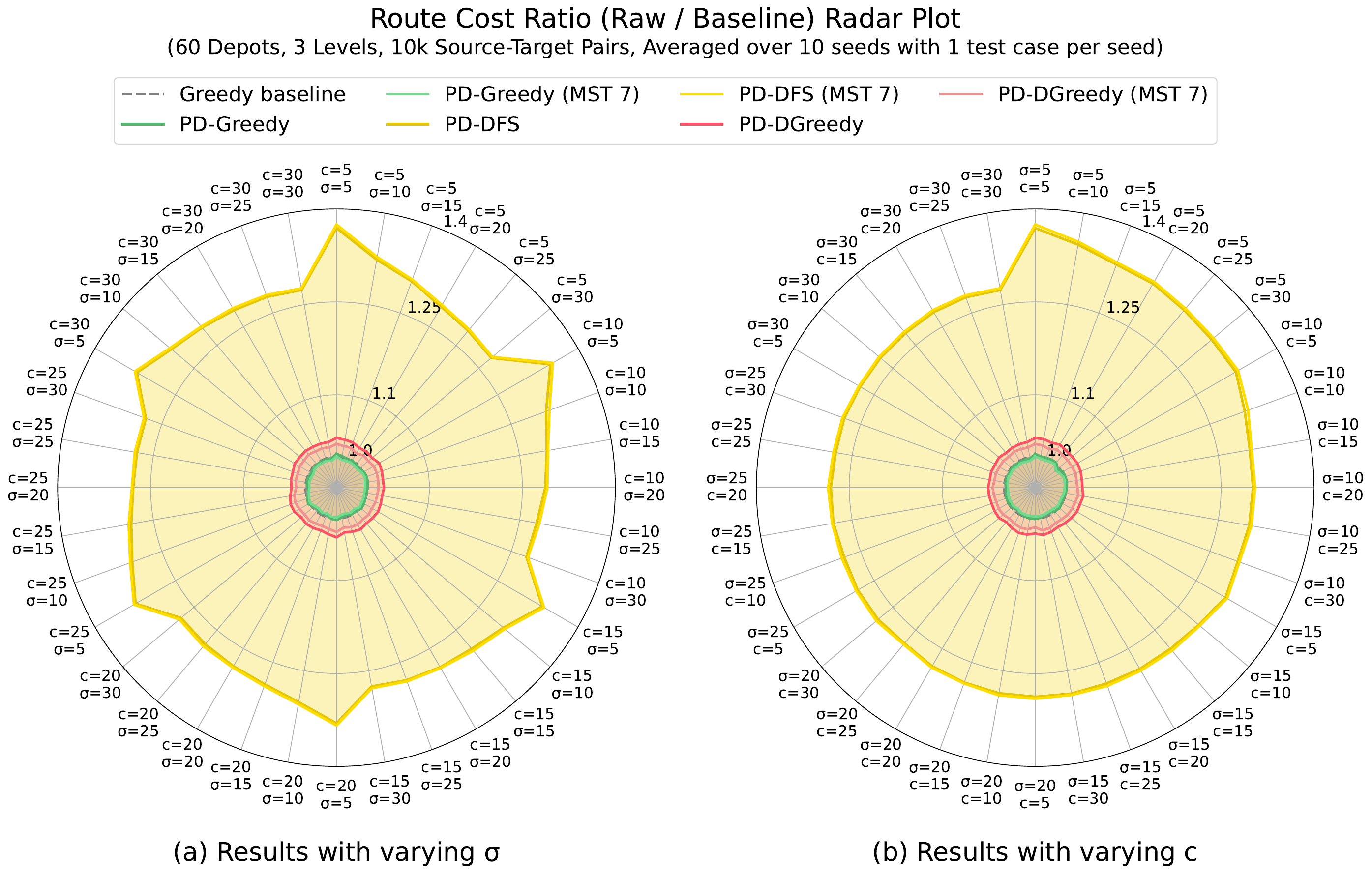}
\caption{Results on the GMM synthetic dataset: varying $c$ and $\sigma$.}  
\label{fig:route_cost_Radar_Plot_two_panels}
\end{figure}

\paragraph{Computational Results in Realworld Data}    
As observed in the synthetic datasets, a relatively small number of depots (vehicles) is sufficient to achieve good performance with low running time. Accordingly, we focus on representative settings with different numbers of depots \(k\), selected at random, and compare the performance of the three PD-based algorithms against the baseline algorithm on both weekday and weekend real-world data, considering both the full set of depots (only for the baseline algorithm) and randomly selected depot subsets.

From Figure~\ref{fig:routeCost_VS_runningTime_ratio_plot_diff_depots_meituan-weekday-18} (weekday data) and Figure~\ref{fig:routeCost_VS_runningTime_ratio_plot_diff_depots_meituan-weekday-22} (weekend data), we observe trends consistent with those seen in the synthetic experiments. 
First, PD-Greedy achieves the best solution quality, while PD-DFS is the fastest among the PD-based methods. 
Second, as the number of depots increases, the total route cost of the PD-based algorithms—particularly PD-greedy and PD-DGreedy—slightly increases from their initially strong performance; however, their cost remains consistently low relative to the baseline, with a ratio no greater than \(1.1\). 
Meanwhile, the running time of the PD-Greedy and PD-DGreedy algorithms continues to decrease as more depots are included, due to the reduced number of packages assigned to each depot. This trend becomes more pronounced for large-scale real-world source--target request pairs.

\begin{figure}[H]
\centering
\includegraphics[width=0.85\textwidth]{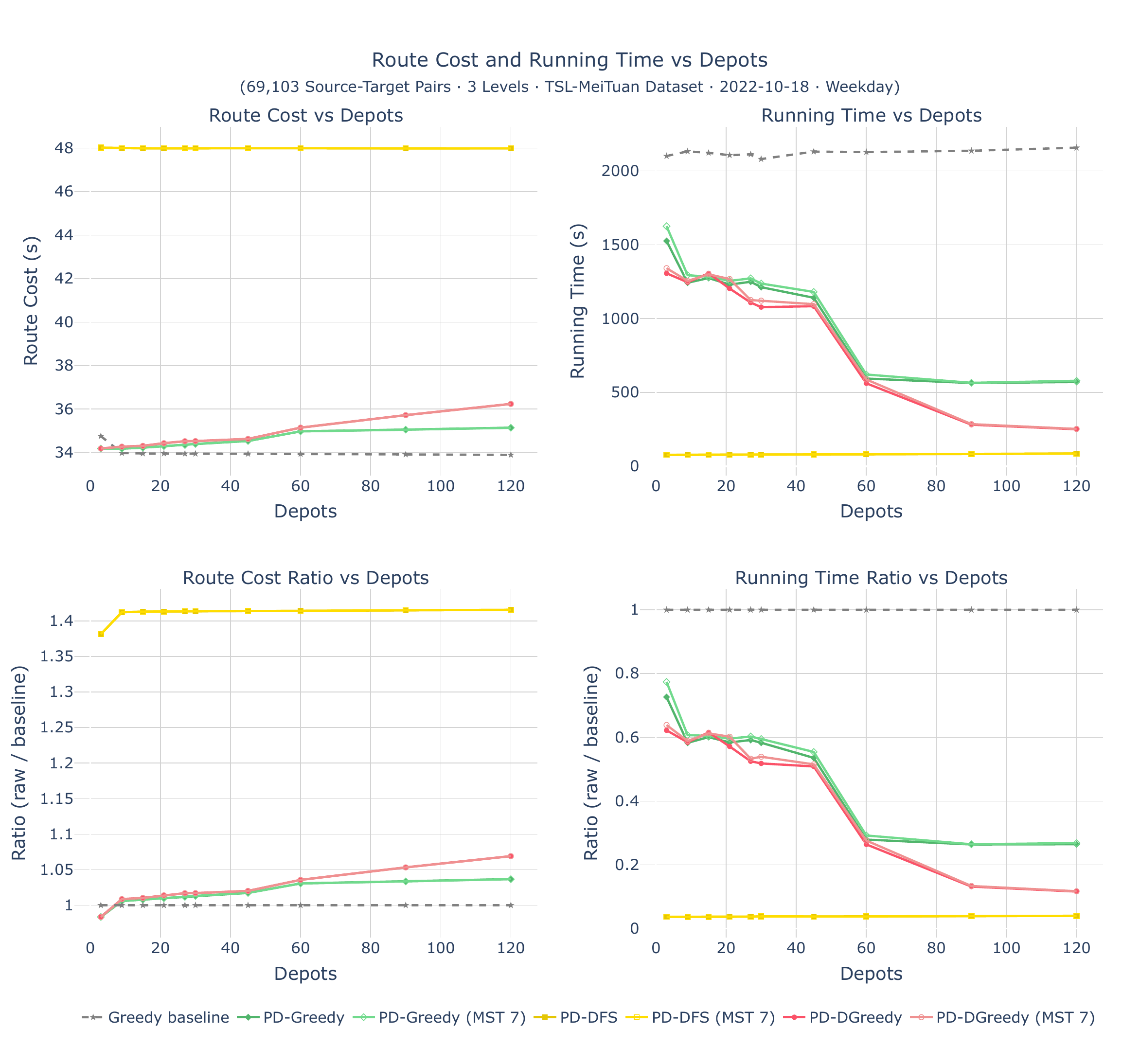}
\caption{Results on the weekday dataset with varying numbers of depots \(k\). 
We compare the PD-based algorithms with the baseline algorithm using the same randomly selected depot sets. 
For reference, when the baseline algorithm uses all depots, the total route cost is \(32.066\) seconds and the running time is \(2137.99\) seconds.}
\label{fig:routeCost_VS_runningTime_ratio_plot_diff_depots_meituan-weekday-18}
\end{figure}

\begin{figure}[H]
\centering
\includegraphics[width=0.85\textwidth]{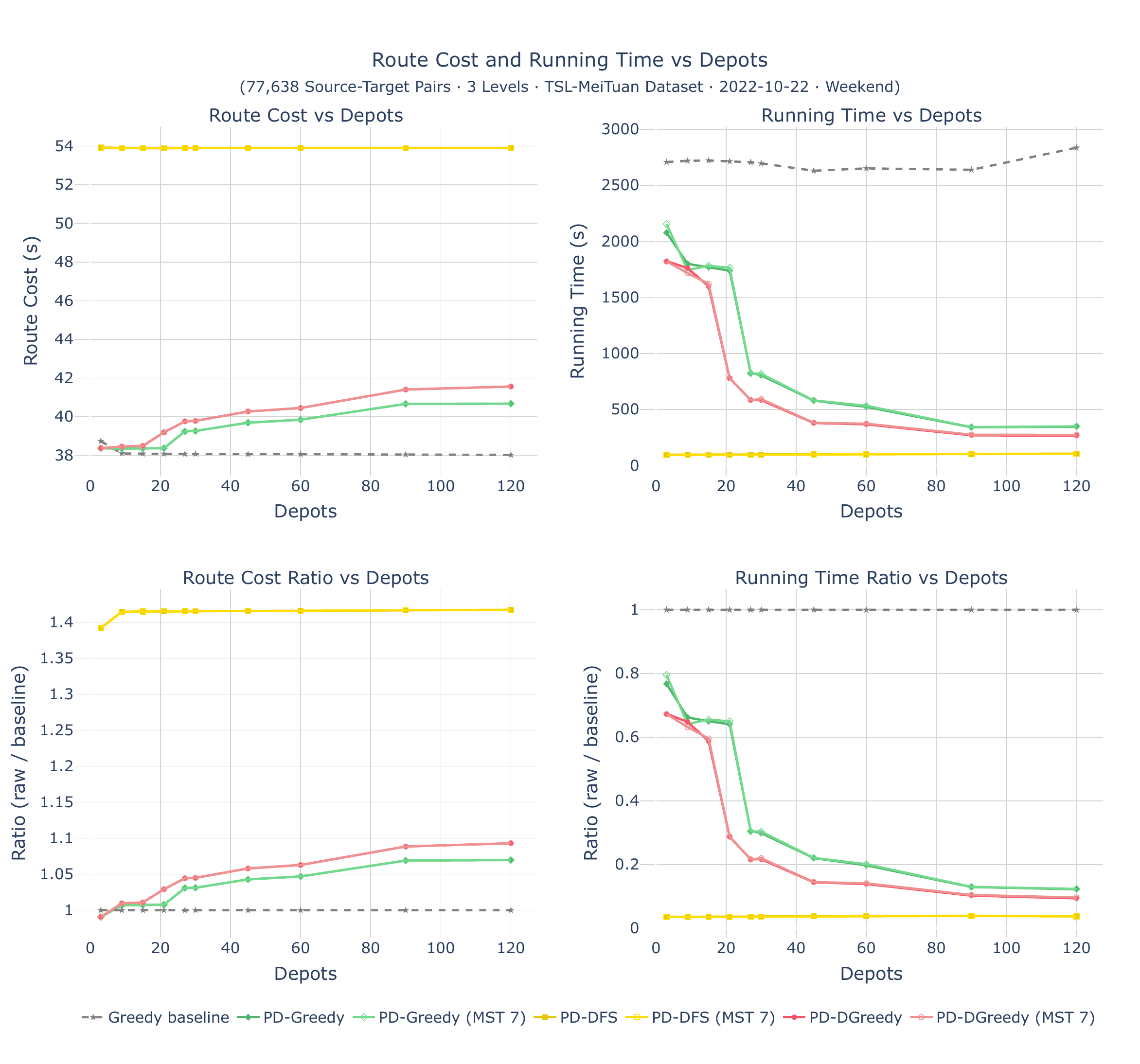}
\caption{Results on the weekend dataset with varying numbers of depots \(k\). 
We compare the PD-based algorithms with the baseline algorithm using the same randomly selected depot sets. 
For reference, when the baseline algorithm uses all depots, the total route cost is \(36.12\) seconds and the running time is \(2731.58\) seconds.}
\label{fig:routeCost_VS_runningTime_ratio_plot_diff_depots_meituan-weekday-22}
\end{figure}

We remark that, for both the synthetic datasets and the real-world datasets, the effects of varying \(n\), \(k\), and \(h\) exhibit similar trends. As the number of request pairs increases, our algorithm continues to achieve favorable performance, particularly when fewer depots are used, while maintaining relatively low running times.

\section{Conclusion}\label{sec:conclusion}

In this paper, we extend the study of multi-package collaborative delivery to the speed-heterogeneous setting and provide new theoretical insights and practical algorithmic approaches. We show that the benefit of fine-grained collaboration over coarse-grained collaboration is limited and devise a primal-dual based constant-factor approximation algorithm for minimizing the total travel time with coarse-grained collaboration. Furthermore, we propose several algorithmic adjustments to improve the performance of the algorithm in practical settings, both in terms of running time and solution quality. In experiments on both synthetic datasets and real-world datasets, we compare our approach
with a cheapest-insertion heuristic as the baseline. The experimental results demonstrate that our approach has significantly improved running-time while achieving comparable or sometimes significantly better solution quality, with the added benefit of having a
provable worst-case performance guarantee.

 \section*{Acknowledgments}
 This research was supported by data provided by Meituan.
 
\bibliographystyle{plainnat}
\bibliography{citepapers}

\appendix

\section{Experimental Results on SY-GMM}
\label{app:SY_GMM_results}

In this appendix, we report additional experimental results on a synthetic dataset generated using a Gaussian Mixture Model, referred to as SY-GMM. This synthetic benchmark is used to evaluate the performance of the proposed methods under controlled and varied conditions, complementing the experiments presented in the main part of the paper.

Figures~\ref{fig:routeCost_VS_runningTime_ratio_plot_diff_stPairs_GMM_5-15}, 
\ref{fig:routeCost_VS_runningTime_ratio_plot_diff_depots_GMM_5-15_100seed}, and 
\ref{fig:routeCost_VS_runningTime_ratio_plot_diff_levels_GMM_5-15} summarize the experimental results on the SY-GMM dataset. Overall, the results exhibit trends consistent with those observed on the synthetic uniform dataset. These findings indicate that the proposed methods are robust to variations in spatial request distribution and scale favorably with increasing problem size, reinforcing the conclusions of this work.

\begin{figure}[H]
\centering
\includegraphics[width=0.85\textwidth]{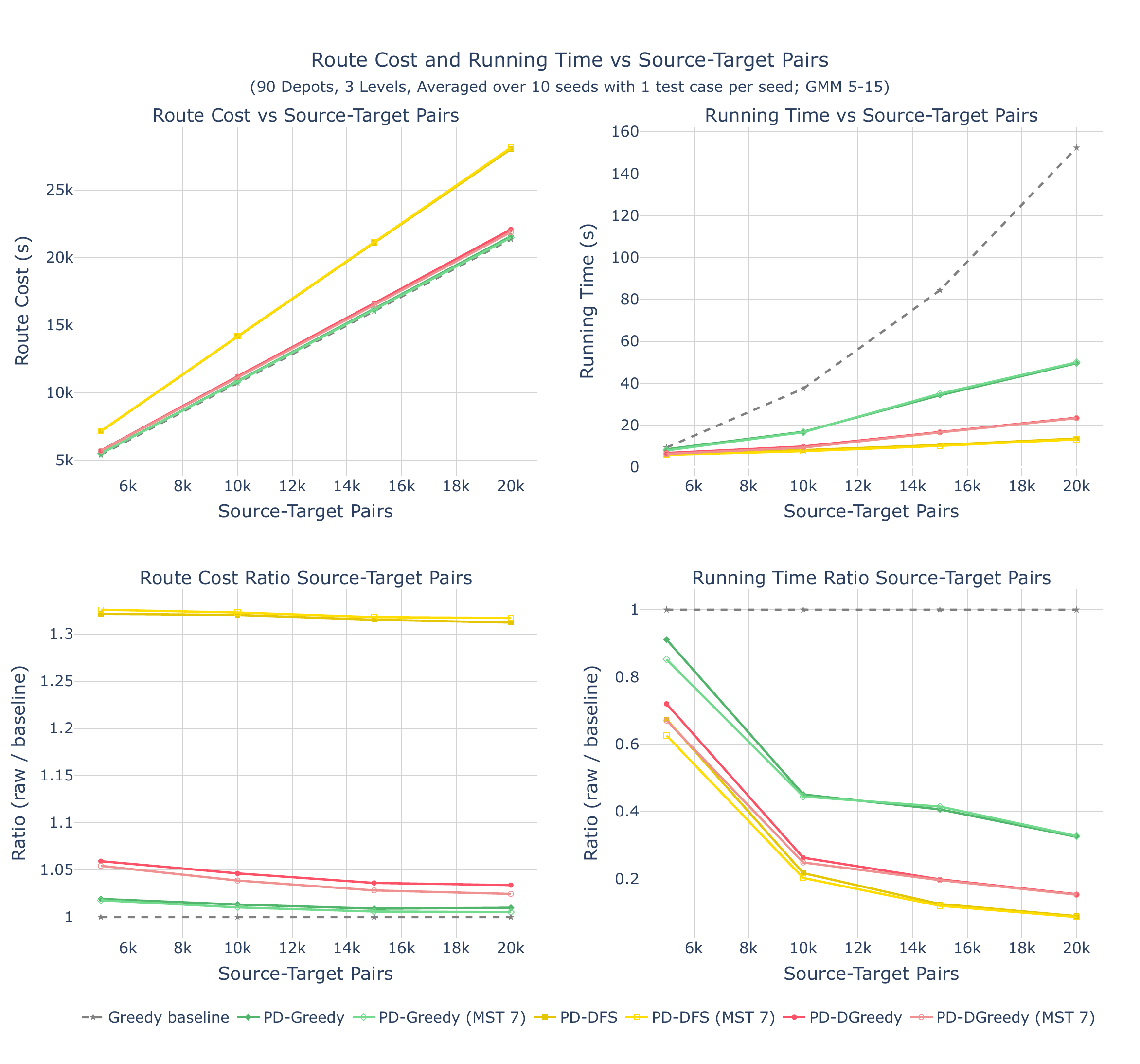}
\caption{Results on the SY-GMM  with varying $n$.}  
\label{fig:routeCost_VS_runningTime_ratio_plot_diff_stPairs_GMM_5-15}
\end{figure}

\begin{figure}[H]
\centering
\includegraphics[width=0.85\textwidth]{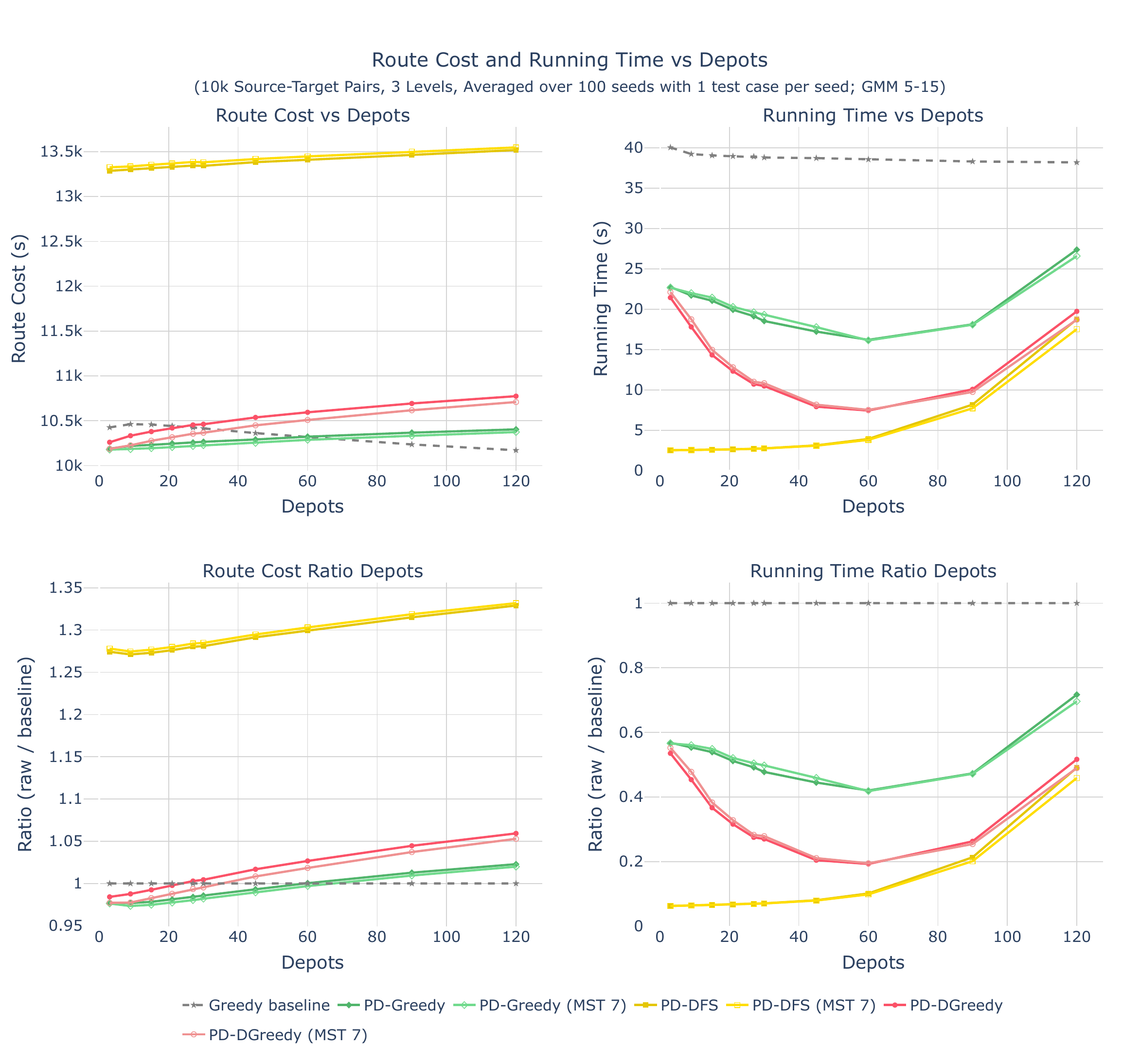}
\caption{Results on the SY-GMM with varying $k$.}  
\label{fig:routeCost_VS_runningTime_ratio_plot_diff_depots_GMM_5-15_100seed}
\end{figure}

\begin{figure}[H]
\centering
\includegraphics[width=0.85\textwidth]{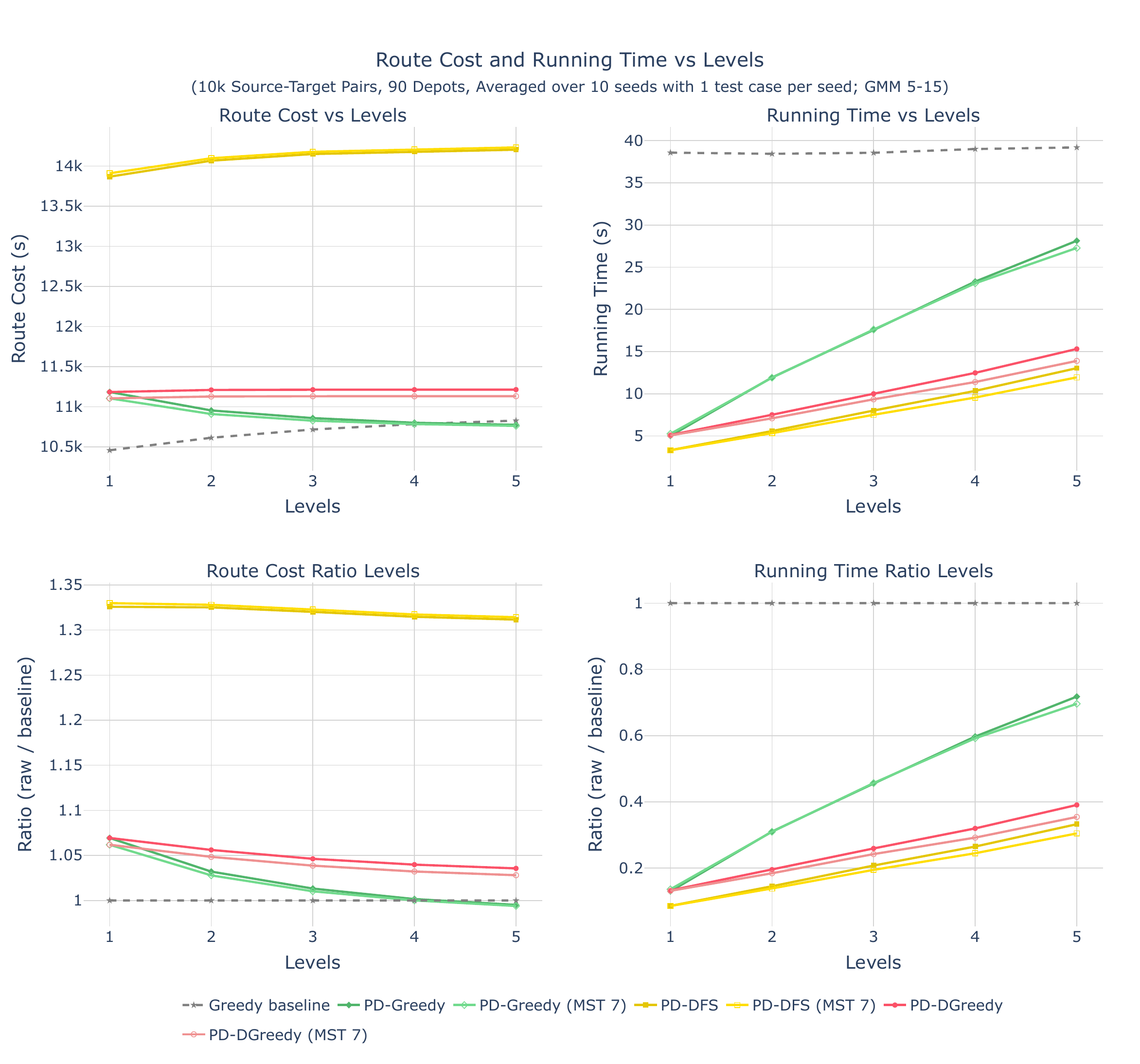}
\caption{Results on the SY-GMM with varying $h$.}  
\label{fig:routeCost_VS_runningTime_ratio_plot_diff_levels_GMM_5-15}
\end{figure}  

\end{document}